\documentclass[
a4paper,
]{article}

\usepackage
[
]
{geometry}

\usepackage{setspace}
\usepackage{Settings}

\usepackage{lineno}

\title{Seven kinds of equivalent models \\ for generalized coalition logics\footnote{
This paper is an extension of the conference paper \cite{chen_each_2025}.
}}

\author{
Zixuan Chen${}^{1,2}$ and Fengkui Ju${}^{1,3}$\footnote{Corresponding author} \vspace{5pt} \\
{\small {$^1$School of Philosophy, Beijing Normal University, China}} \vspace{2.5pt} \\
{\small {$^2$\href{mailto:zixuan.chen21@outlook.com}{zixuan.chen21@outlook.com}}} \vspace{2.5pt} \\
{\small {$^3$\href{mailto:fengkui.ju@bnu.edu.cn}{fengkui.ju@bnu.edu.cn}}}
}

\date{}

\begin{document}

\maketitle

\setlength{\parskip}{0.5em}


\begin{abstract}

\noindent Coalition Logic is an important logic in logical research on strategic reasoning.
In two recent papers, Li and Ju argued that generally, concurrent game models, models of Coalition Logic, have three too strong assumptions: seriality, independence of agents, and determinism.
They presented eight coalition logics based on eight classes of general concurrent game models, determined by which of the three assumptions are met.
In this paper, we show that each of the eight sets of valid formulas of the eight logics is determined by six other kinds of models, that is, single-coalition-first action models, single-coalition-first actual neighborhood models, clear grand-coalition-first action models, clear single-coalition-first actual neighborhood models, tree-like grand-coalition-first action models, and tree-like single-coalition-first actual neighborhood models.

\medskip


\end{abstract}

\section{Introduction}
\label{sec:Introduction}

\subsection{Coalition Logic and its two kinds of semantics}

Coalition Logic $\FCL$, proposed by Pauly \cite{pauly_modal_2002}, is an important logic in logical research on strategic reasoning.
Many logics in this field are extensions of $\FCL$. For example, Alternating-time Temporal Logic $\FATL$ \cite{alur_alternating-time_2002} is a temporal extension of $\FCL$, and Strategy Logic $\FSL$ \cite{mogavero_reasoning_2014} is an extension of $\FCL$, whose language has quantifiers for and names of strategies. We refer to \cite{benthem_models_2015}, \cite{agotnes_knowledge_2015} and \cite{sep-logic-power-games} for overviews of the area.

The language of $\FCL$ is a modal language with the featured operator $\Fclo{\FCC} \phi$, indicating \emph{some available joint action of coalition $\FCC$ ensures $\phi$}.
$\FCL$ has two kinds of equivalent semantics~\cite{pauly_modal_2002, goranko_strategic_2013}: action semantics and neighborhood semantics.

In action semantics, models are \emph{concurrent game models}. Roughly, in a concurrent game model:
there are some states;
there are some agents, who can form coalitions;
for every coalition, there is an availability function, specifying available joint actions of the coalition at states;
for every coalition, there is an outcome function, specifying possible outcome states of joint actions of the coalition.
The formula $\Fclo{\FCC} \phi$ is true at a state in a concurrent game model if $\FCC$ has an available joint action such that $\phi$ is true at every possible outcome state of the action.

In neighborhood semantics, models are \emph{alpha neighborhood models}. Alpha neighborhood models do not have actions. Consequently, they do not have availability functions and outcome functions. Instead, for every coalition, there is a neighborhood function that specifies the alpha powers of the coalition.
An alpha power is a set of states into which the coalition can force the next moment to fall~\cite{moulin_cores_1982}.
The formula $\Fclo{\FCC} \phi$ is true at a state in an alpha neighborhood model if $\FCC$ has an alpha power such that $\phi$ is true at every state of the power.

Concurrent game models and alpha neighborhood models of $\FCL$ have their respective advantages and disadvantages. Concurrent game models are more intuitive. However, they contain information for which the language of $\FCL$ is blind, that is, actions. Neighborhood models are simpler: they only keep the needed information for evaluating formulas of $\FCL$. Some work in the line of $\FCL$, such as \cite{agotnes_coalition_2018}, uses neighborhood models

\subsection{Eight coalition logics determined by eight classes of general concurrent game models}

Concurrent game models have three assumptions, among others.
The first one is \emph{seriality}: every coalition always has an available joint action.
The second one is \emph{the independence of agents}: the merge of two available joint actions of two disjoint coalitions is always an available joint action of the union of the two coalitions.
The third one is \emph{determinism}: every available joint action of the grand coalition has a unique outcome. The three assumptions are also reflected in the alpha neighborhood models of $\FCL$.

In the agency literature, some authors have used models in which one or more of the three assumptions do not hold.
In the tradition of coalition logic, Jiang and Naumov \cite{JIANG2022103727} used non-deterministic concurrent game models. 
In the tradition of STIT (see to it that) logic, Sergot \cite{sergot_examples_2014, sergot_actual_2022} argued that neither the independence of agents nor determinism commonly holds.  
At the intersection of $\FATL$ and STIT logic, Boudou and Lorini \cite{10.5555/3237383.3237443} utilized non-deterministic concurrent game models.
In the tradition of Propositional Dynamic Logic, Royakkers and Hughes \cite{royakkers_blame_2020} argued that the independence of agents does not hold generally and used models in which the independence of agents and determinism fail.

Against the agency literature, Li and Ju \cite{li_minimal_2025} argued that the three assumptions are too strong.
Here, we briefly mention their arguments.
Usually, games terminate when reaching specific states. For example, a paper-scissors-rock game terminates when a winner is decided. Intuitively, in terminating states, players have no available actions. Thus, seriality might not hold.
There are many situations in which whether a coalition can perform an action is conditional on other agents' actions at the same time. Here is an example. There are two agents in a room, $a$ and $b$, and only one chair. Agent $a$ can sit, and agent $b$ can sit, but they cannot sit at the same time. The independence of agents fails in these situations.
In many situations, some joint actions of all behaving agents have more than one outcome state. The following example, from \cite{sergot_examples_2014}, can illustrate this.
A vase stands on a table. There is an agent $a$ who can lift or lower the end of the table. If the table tilts, the vase might fall, and if it falls, it might break. Therefore, determinism might not hold.

Based on \emph{general concurrent game models}, which do not have the three assumptions, Li and Ju \cite{li_minimal_2025} proposed a Minimal Coalition Logic $\FMCL$.
Although the three assumptions are generally too strong, we might want to keep some of them when constructing logics for strategic reasoning in some situations. Considering which of the three properties we want to keep, there are eight coalition logics in total. Li and Ju \cite{li_completeness_2024} showed the completeness of the eight logics in a uniform way.

\subsection{Our work}

The eight coalition logics have eight sets of valid formulas. In this work, we show the following: for each of the eight sets of valid formulas, (1) it is also determined by three other classes of \emph{action models}, and (2) it is determined by three other classes of \emph{actual neighborhood models}.
These different kinds of models make good sense for representing coalitional powers. Our work makes it clear what the logics with respect to these kinds of models are.

Action models are those like concurrent game models. In actual neighborhood models, for every coalition, there is a neighborhood function that specifies \emph{actual powers} of the coalition.
An actual power is a set of states such that the coalition has a choice such that (1) the choice forces the next moment to fall into the set, and (2) every state in the set is compatible with the choice \cite{benthem_new_2019}.

To ease understanding the differences among the six kinds of models, in this paper, we call general concurrent game models \emph{grand-coalition-first action models}. The reason is that in general concurrent game models, the outcome functions of all coalitions can be derived from the outcome functions of the grand coalition.

The six kinds of models are as follows:
\begin{enumerate}

\item 

\begin{enumerate}

\item 

The first kind is \emph{single-coalition-first action models}. In this kind of models, the outcome functions for single coalitions determine those for other coalitions.
This kind of models makes good sense for situations where every agent can change a component of the state, and different components of the state that can be changed by different agents are independent of each other.

\item 

The second kind is \emph{single-coalition-first actual neighborhood models}. In this kind of models, the neighborhood functions for single coalitions determine those for other coalitions.

\end{enumerate}

\item 

\begin{enumerate}

\item 

The third kind is \emph{clear grand-coalition-first action models}. In this kind of models, different joint actions of the grand coalition have different outcomes. Clear grand-coalition-first action models are quite common in game theory.

\item 

The fourth kind is \emph{clear single-coalition-first actual neighborhood models}. In this kind of models, the actual powers of an agent are pairwise disjoint.

\end{enumerate}

\item 

\begin{enumerate}

\item 

The fifth kind is \emph{tree-like grand-coalition-first action models}. In this kind of models, every state has a determined history.

\item 

The sixth kind is \emph{tree-like single-coalition-first actual neighborhood models}. In this kind of models, every state has a determined history, too.

\end{enumerate}

\end{enumerate}

This work offers new models for logics like Alternating-time Temporal Logic or STIT logic, which will be our future work.

\paragraph{Structure of the paper}

\begin{itemize}
    
\item 

In Section \ref{section:General settings of coalition logics}, we present some general settings of coalition logics, including their language, their action semantics and neighborhood semantics, and the transformation from action semantics to neighborhood semantics.

\item 

In Section \ref{section:Eight coalition logics determined by eight classes of grand-coalition-first action models}, we introduce the eight coalition logics determined by eight classes of grand-coalition-first action models.

\item 

In Section \ref{section:Single-coalition-first action models and single-coalition-first neighborhood models}, we define single-coalition-first action models and show that they are grand-coalition-first action models. 

We define single-coalition-first actual neighborhood models and show that they can \emph{represent} single-coalition-first action models.

\item 

In Section \ref{section:Clear grand-coalition-first action models and clear single-coalition-first neighborhood models}, we define clear grand-coalition-first action models and show that they are single-coalition-first action models.

We define clear single-coalition-first actual neighborhood models and show that they can represent clear grand-coalition-first action models.

\item 

In Section \ref{section:Tree-like grand-coalition-first action models and tree-like single-coalition-first neighborhood models}, we define tree-like grand-coalition-first action models and show that they are clear grand-coalition-first action models.

We define tree-like single-coalition-first actual neighborhood models and show that they can represent tree-like grand-coalition-first action models.

\item 

In Section \ref{section:Each of the eight coalition logics is determined by the six kinds of models, too}, we show that each of the eight coalition logics is determined by the six kinds of models with respective properties.

\item 

We conclude in Section \ref{section:Conclusion}.

\end{itemize}

\newcommand{\ja}[1]{\sigma_{{#1}}}

\newcommand{\FAE}{\mathtt{AE}}
\newcommand{\FLE}{\mathtt{LE}}

\newcommand{\FYY}{\mathrm{Y}}

\newcommand{\AM}{\mathtt{AM}}
\newcommand{\FAM}{\mathtt{AM}}

\newcommand{\NM}{\mathtt{NM}}
\newcommand{\FNM}{\mathtt{NM}}

\newcommand{\GAM}{\mathtt{G}\textbf{-}\mathtt{AM}}
\newcommand{\FGAM}{\mathtt{G}\textbf{-}\mathtt{AM}}

\newcommand{\SAM}{\mathtt{S}\textbf{-}\mathtt{AM}}
\newcommand{\FSAM}{\mathtt{S}\textbf{-}\mathtt{AM}}

\newcommand{\SNM}{\mathtt{S}\textbf{-}\mathtt{NM}}
\newcommand{\FSNM}{\mathtt{S}\textbf{-}\mathtt{NM}}

\newcommand{\FES}{\mathrm{ES}}

\newcommand{\Fnei}{\mathtt{nei}}

\newcommand{\Fsuc}{\Fred{\mathtt{suc}}}

\section{General settings of coalition logics}
\label{section:General settings of coalition logics}

\subsection{Language}

Let $\FAP$ be a countable set of atomic propositions, and $\FAG$ be a finite \emph{nonempty} set of agents. The subsets of $\FAG$ are called \Fdefs{coalitions}, and $\FAG$ is called the \Fdefs{grand coalition}.
In the sequel, when no confusion arises, we often write $a$ instead of $\{a\}$, where $a \in \FAG$.

\begin{definition}[The language $\Phi$]
\label{definition:The language Phi CL}

The language $\Phi$ is defined as follows, where $p$ ranges over $\FAP$ and $\FCC \subseteq \FAG$:
\[
\phi ::=\top \mid p \mid \neg \phi \mid (\phi \wedge \phi) \mid \Fclo{\FCC} \phi
\]

\end{definition}

The formula $\Fclo{\FCC} \phi$ indicates \emph{some available joint action of $\FCC$ ensures $\phi$}. Here, the notation of modalities differs from their notation in the literature. We do this to indicate the iteration of quantifiers in their meaning.
The propositional connectives $\bot, \lor, \rightarrow$ and $\leftrightarrow$ are defined as usual.

\subsection{Action semantics}

Let $\FAC$ be a \emph{nonempty} set of actions. For every $\FCC \subseteq \FAG$, we define $\FJA_\FCC = \{\sigma_\FCC \mid \sigma_\FCC: \FCC \rightarrow \FAC\}$, which is a set of \Fdefs{joint actions} of $\FCC$. Joint actions of the grand coalition are called \Fdefs{action profiles}. In the sequel, we sometimes use sequences of actions to represent joint actions of coalitions, where an order among agents is implicitly presupposed. We define $\FJA = \bigcup \{\FJA_\FCC \mid \FCC \subseteq \FAG\}$.

For every $\FCC, \FDD \subseteq \FAG$ such that $\FDD \subseteq \FCC$ and $\sigma_\FCC \in \FJA_\FCC$, we use $\sigma_\FCC |_\FDD$ to indicate the \emph{restriction} of $\sigma_\FCC$ to $\FDD$, which is in $\FJA_\FDD$.

\begin{definition}[Action models]
\label{definition:??}

An \Fdefs{action model} is a tuple $\FAM = (\FST, \FAC, \{\Fout_\FCC \mid \FCC \subseteq \FAG\}, \{\Faja_\FCC \mid \FCC \subseteq \FAG\}, \Flab)$, where:
\begin{itemize}

\item

$\FST$ is a nonempty set of states.

\item

$\FAC$ is a nonempty set of actions.

\item

for every $\FCC \subseteq \FAG$, $\Fout_\FCC: \FST \times \FJA_\FCC \rightarrow \mathcal{P}(\FST)$ is an \Fdefs{outcome function} for $\FCC$.

\item

for every $\FCC \subseteq \FAG$, $\Faja_\FCC: \FST \rightarrow \mathcal{P} (\FJA_\FCC)$ is an \Fdefs{availability function} for $\FCC$ meeting the following condition: for all $s \in \FST$, $\Faja_\FCC (s) = \{\ja{\FCC} \in \FJA_\FCC \mid \Fout_\FCC (s, \ja{\FCC}) \neq \emptyset\}$.

\item

$\Flab: \FST \rightarrow \mathcal{P}(\FAP)$ is a \Fdefs{labeling function}.

\end{itemize}
\end{definition}


Different coalition logics might put various constraints on action models.

\begin{definition}[Action semantics]
\label{definition:??}

Let $\FAM = (\FST, \FAC, \{\Fout_\FCC \mid \FCC \subseteq \FAG\}, \{\Faja_\FCC \mid \FCC \subseteq \FAG\}, \Flab)$ be an action model. Then truth conditions of formulas in $\Phi$ at states $s$
are defined recursively as follows:

\begin{tabular}{lll}
$\FAM, s \Vdash \top$ & & \\
$\FAM, s \Vdash p$ & $\Leftrightarrow$ & \parbox[t]{28em}{$p \in \Flab (s)$} \\
$\FAM, s \Vdash \neg \phi$ & $\Leftrightarrow$ & \parbox[t]{28em}{not $\FAM, s \Vdash \phi$} \\
$\FAM, s \Vdash \phi \land \psi$ & $\Leftrightarrow$ & \parbox[t]{28em}{$\FAM, s \Vdash \phi$ and $\FAM, s \Vdash \psi$} \\
$\FAM, s \Vdash \Fclo{\FCC} \phi$ & $\Leftrightarrow$ & \parbox[t]{28em}{there is $\ja{\FCC} \in \Faja_\FCC (s)$ such that for all $t \in \Fout_\FCC (s, \ja{\FCC})$, $\FAM, t \Vdash \phi$}
\end{tabular}

\end{definition}

\subsection{Neighborhood semantics}

\begin{definition}[Actual neighborhood models]
\label{definition:??}

An \Fdefs{actual neighborhood model} is a tuple $\FNM = (\FST, \{\Fnei_\FCC \mid \FCC \subseteq \FAG\}, \Flab)$, where:
\begin{itemize}

\item

$\FST$ is a nonempty set of states.

\item

for every $\FCC \subseteq \FAG$, $\Fnei_\FCC: \FST \rightarrow \mathcal{P}(\mathcal{P}(\FST))$ is a \Fdefs{neighborhood function} for $\FCC$.

\item

$\Flab: \FST \rightarrow \mathcal{P}(\FAP)$ is a labeling function.

\end{itemize}

\end{definition}

Intuitively, $\Fnei_\FCC (s)$ consists of \Fdefs{actual powers} of $\FCC$ at $s$.


\begin{definition}[Alpha neighborhood models]
\label{definition:??}

An \Fdefs{alpha neighborhood model} is a tuple $\FNM^\alpha = (\FST, \{\Fnei^\alpha_\FCC \mid \FCC \subseteq \FAG\}, \Flab)$, where:
\begin{itemize}

\item

$\FST$ is a nonempty set of states.

\item

for every $\FCC \subseteq \FAG$, $\Fnei^\alpha_\FCC: \FST \rightarrow \mathcal{P}(\mathcal{P}(\FST))$ is a \Fdefs{neighborhood function} for $\FCC$ such that for every $s \in \FST$, $\Fnei^\alpha_\FCC (s)$ is closed under supersets.

\item

$\Flab: \FST \rightarrow \mathcal{P}(\FAP)$ is a labeling function.

\end{itemize}

\end{definition}

Intuitively, $\Fnei^\alpha_\FCC (s)$ consists of \Fdefs{alpha powers} of $\FCC$ at $s$.

Different coalition logics might put various constraints on neighborhood models.

\begin{definition}[Neighborhood semantics]
\label{definition:??}

Let $\FNM = (\FST, \{\Fnei_\FCC \mid \FCC \subseteq \FAG\}, \Flab)$ be an actual/alpha neighborhood model. Then truth conditions of formulas in $\Phi$ at states $s$
are defined recursively as follows:

\begin{tabular}{lll}
$\FNM, s \Vdash \top$ & & \\
$\FNM, s \Vdash p$ & $\Leftrightarrow$ & \parbox[t]{28em}{$p \in \Flab (s)$} \\
$\FNM, s \Vdash \neg \phi$ & $\Leftrightarrow$ & \parbox[t]{28em}{not $\FNM, s \Vdash \phi$} \\
$\FNM, s \Vdash \phi \land \psi$ & $\Leftrightarrow$ & \parbox[t]{28em}{$\FNM, s \Vdash \phi$ and $\FNM, s \Vdash \psi$} \\
$\FNM, s \Vdash \Fclo{\FCC} \phi$ & $\Leftrightarrow$ & \parbox[t]{28em}{there is $\FYY \in \Fnei_\FCC (s)$ such that for all $t \in \FYY$, $\FNM, t \Vdash \phi$}
\end{tabular}

\end{definition}

\subsection{Transformation of action semantics to neighborhood semantics}

\begin{definition}[Actual effectivity functions and alpha effectivity functions of action models]
\label{definition:??}

Let $\FAM = (\FST, \FAC, \{\Fout_\FCC \mid \FCC \subseteq \FAG\}, \{\Faja_\FCC \mid \FCC \subseteq \FAG\}, \Flab)$ be an action model.

\begin{itemize}

\item 

For every $\FCC \subseteq \FAG$, the function $\FAE_\FCC$ defined as follows is called the \Fdefs{actual effectivity function} for $\FCC$ in $\FAM$: for every $s \in \FST$,
\[
\FAE_\FCC (s) = \{\Fout_\FCC (s, \sigma_\FCC) \mid \sigma_\FCC \in \Faja_\FCC (s)\}.
\]

\item

For every $\FCC \subseteq \FAG$, the function $\FLE_\FCC$ defined as follows is called the \Fdefs{alpha effectivity function} for $\FCC$ in $\FAM$: for every $s \in \FST$,
\[
\FLE_\FCC (s) = \{\FYY \subseteq \FST \mid \Fout_\FCC (s, \sigma_\FCC) \subseteq \FYY \text{ for some } \sigma_\FCC \in \Faja_\FCC (s)\}.
\]

\end{itemize}

\end{definition}

Actual effectivity functions are more transparent than alpha effectivity functions in the following sense: every element $\FYY$ of $\FAE_\FCC (s)$ can be understood as a class of equivalent actions such that $\FYY$ is the set of their possible outcome states at $s$.

The following fact, which is easy to verify, indicates that we can transform an action model to an actual/alpha neighborhood model without changing the satisfaction of formulas:

\begin{fact}[]
\label{fact:actual alpha effectivity functions}

Let $\FAM = (\FST, \FAC, \{\Fout_\FCC \mid \FCC \subseteq \FAG\}, \{\Faja_\FCC \mid \FCC \subseteq \FAG\}, \Flab)$ be an action model.
\begin{itemize}

\item 

Then $\FNM = (\FST, \{\FAE_\FCC \mid \FCC \subseteq \FAG\}, \Flab)$ is an actual neighborhood model such that for every state $s$ of $\FAM$ and formula $\phi$ in $\Phi$, $\FAM,s \Vdash \phi$ if and only if $\FNM,s \Vdash \phi$.

\item

Then $\FNM^\alpha = (\FST, \{\FLE_\FCC \mid \FCC \subseteq \FAG\}, \Flab)$ is an alpha neighborhood model such that for every state $s$ of $\FAM$ and formula $\phi$ in $\Phi$, $\FAM,s \Vdash \phi$ if and only if $\FNM^\alpha,s \Vdash \phi$.

\end{itemize}

\end{fact}

\begin{definition}[Action models $z$-representable by actual neighborhood models and $\alpha$-representable by alpha neighborhood models]
\label{definition:??}
~

\begin{itemize}

\item 

An action model $\FAM = (\FST, \FAC, \{\Fout_\FCC \mid \FCC \subseteq \FAG\}, \{\Faja_\FCC \mid \FCC \subseteq \FAG\}, \Flab)$ is \Fdefs{$z$-representable} by an actual neighborhood model $\FNM = (\FST, \{\Fnei_\FCC \mid \FCC \subseteq \FAG\}, \Flab)$ if for every $\FCC \subseteq \FAG$, $\FAE_\FCC = \Fnei_\FCC$.

A class of action models $\mathbf{AM}$ is \Fdefs{$z$-representable} by a class of actual neighborhood models $\mathbf{NM}$ if (1) every action model in $\mathbf{AM}$ is $z$-representable by an actual neighborhood model in $\mathbf{NM}$, and (2) every actual neighborhood model in $\mathbf{NM}$ $z$-represents an action model in $\mathbf{AM}$.

\item

An action model $\FAM = (\FST, \FAC, \{\Fout_\FCC \mid \FCC \subseteq \FAG\}, \{\Faja_\FCC \mid \FCC \subseteq \FAG\}, \Flab)$ is \Fdefs{$\alpha$-representable} by an alpha neighborhood model $\FNM^\alpha = (\FST, \{\Fnei^\alpha_\FCC \mid \FCC \subseteq \FAG\}, \Flab)$ if for every $\FCC \subseteq \FAG$, $\FLE_\FCC = \Fnei^\alpha_\FCC$.

A class of action models $\mathbf{AM}$ is \Fdefs{$\alpha$-representable} by a class of alpha neighborhood models $\mathbf{NM}^\alpha$ if (1) every action model in $\mathbf{AM}$ is $\alpha$-representable by an alpha neighborhood model in $\mathbf{NM}^\alpha$, and (2) every alpha neighborhood model in $\mathbf{NM}^\alpha$ $\alpha$-represents an action model in $\mathbf{AM}$.

\end{itemize}

\end{definition}

The following result can be easily shown:

\begin{theorem}[Representability implies equivalence]
\label{theorem:classes representable}
~

\begin{enumerate}[label=(\arabic*),leftmargin=3.33em]

\item

For every class of action models $\mathbf{AM}$ and class of neighborhood models $\mathbf{NM}$, if $\mathbf{AM}$ is $z$-representable by $\mathbf{NM}$, then for every formula $\phi$ in $\Phi$, $\phi$ is valid with respect to $\mathbf{AM}$ if and only if $\phi$ is valid with respect to $\mathbf{NM}$.

\item 

For every class of action models $\mathbf{AM}$ and class of alpha neighborhood models $\mathbf{NM}^\alpha$, if $\mathbf{AM}$ is $\alpha$-representable by $\mathbf{NM}^\alpha$, then for every formula $\phi$ in $\Phi$, $\phi$ is valid with respect to $\mathbf{AM}$ if and only if $\phi$ is valid with respect to $\mathbf{NM}^\alpha$.

\end{enumerate}

\end{theorem}

Assume we have a coalition logic and the class of its models is $\mathbf{AM}$, a class of action models. By this result, if $\mathbf{AM}$ is representable by a class of neighborhood models $\mathbf{NM}$, then we can transform the action semantics of the coalition logic to the neighborhood semantics based on $\mathbf{NM}$.

From Fact \ref{fact:actual alpha effectivity functions}, from every action model in $\FAM$, we can get a neighborhood model $\FNM$ representing $\FAM$. Let $\mathbf{NM}$ be the collection of all these neighborhood models. Clearly, $\mathbf{AM}$ is representable by $\mathbf{NM}$. Note that this transformation is not very useful. To make a transformation meaningful, $\mathbf{NM}$ shall be defined \emph{directly}.

\medskip

\emph{In this paper, we mainly consider actual neighborhood models. In the sequel, by neighborhood models, we mean actual neighborhood models, unless explicitly stated.}

\section{Eight coalition logics determined by eight classes of \\ grand-coalition-first action models}
\label{section:Eight coalition logics determined by eight classes of grand-coalition-first action models}

\subsection{Eight classes of grand-coalition-first action models}

\begin{definition}[Grand-coalition-first action models]
\label{definition:??}

An action model $\FGAM = (\FST, \FAC, \{\Fout_\FCC \mid \FCC \subseteq \FAG\}, \{\Faja_\FCC \mid \FCC \subseteq \FAG\}, \Flab)$ is a \Fdefs{grand-coalition-first action model} if for all $\FCC \subseteq \FAG$, $s \in \FST$ and $\sigma_\FCC \in \FJA_\FCC$: $\Fout_\FCC (s, \ja{\FCC}) = \bigcup \{\Fout_\FAG (s, \ja{\FAG}) \mid \ja{\FAG} \in \FJA_\FAG \text{ and } \ja{\FCC} \subseteq \ja{\FAG}\}$.

\end{definition}

It can be seen that for every $\FCC \subseteq \FAG$, $\Fout_\FCC$ can be derived from $\Fout_\FAG$, which is why we use ``grand-coalition-first''.

\begin{definition}[Seriality, independence, and determinism of grand-coalition-first action models]
\label{definition:Seriality, independence, and determinism of grand-coalition-first action models}

Let $\FGAM = (\FST, \FAC, \{\Fout_\FCC \mid \FCC \subseteq \FAG\}, \{\Faja_\FCC \mid \FCC \subseteq \FAG\}, \Flab)$ be a grand-coalition-first action model.

We say:
\begin{itemize}

\item

$\GAM$ is \Fdefs{serial} if $\Faja_\FCC (s) \neq \emptyset$ for all $\FCC \subseteq \FAG$ and $s \in \FST$;

\item

$\GAM$ is \Fdefs{independent} if $\ja{\FCC} \cup \ja{\FDD} \in \Faja_{\FCC \cup \FDD} (s)$ for all $s\in \FST$, $\FCC, \FDD \subseteq \FAG$ such that $\FCC\cap \FDD = \emptyset$, $\ja{\FCC} \in \Faja_\FCC (s)$, and $\ja{\FDD} \in \Faja_\FDD (s)$;

\item

$\GAM$ is \Fdefs{deterministic} if ${\Fout_\FAG (s, \ja{\FAG})}$ is a singleton for all $s\in \FST$ and $\ja{\FAG}\in \Faja_\FAG (s)$.

\end{itemize}

\end{definition}

We let the three symbols $\mathtt{S}$, $\mathtt{I}$ and $\mathtt{D}$ correspond to the three properties, respectively, and let the eight strings $\epsilon$, $\mathtt{S}$, $\mathtt{I}$, $\mathtt{D}$, $\mathtt{SI}$, $\mathtt{SD}$, $\mathtt{ID}$, and $\mathtt{SID}$ correspond to the eight combinations of the three properties, respectively. We use $\FES$ to indicate the set of the eight strings.

For every $\FXX \in \FES$ and grand-coalition-first action model $\FGAM$, we say $\GAM$ is a \Fdefs{$\FXX$-model} if $\GAM$ has the properties corresponding to $\FXX$.

For every $\FXX \in \FES$, we use $\FXL$ to refer to the coalition logic determined by the class of grand-coalition-first action $\FXX$-models.

\subsection{Eight axiomatic systems}

\begin{definition}[An axiomatic system for $\FMCL$]
~

\noindent \emph{Axioms}:

\vspace{5pt}

\begin{tabular}{rl}
Tautologies ($\mathtt{A}\text{-}\mathtt{Tau}$): & all propositional tautologies \vspace{5pt} \\
No absurd available actions ($\mathtt{A}\text{-}\mathtt{NAAA}$): & $\neg \Fclo{\FCC} \bot$ \vspace{5pt} \\
Monotonicity of goals ($\mathtt{A}\text{-}\mathtt{MG}$): & $\Fclo{\emptyset} (\phi \rightarrow \psi) \rightarrow (\Fclo{\FCC} \phi \rightarrow \Fclo{\FCC} \psi)$ \vspace{5pt} \\
Monotonicity of coalitions ($\mathtt{A}\text{-}\mathtt{MC}$): & $\Fclo{\FCC} \phi \rightarrow \Fclo{\FDD} \phi$, where $\FCC \subseteq \FDD$
\end{tabular}

\vspace{5pt}

\noindent \emph{Inference rules}:

\vspace{5pt}

\begin{tabular}{rl}
Modus ponens ($\mathtt{MP}$): & $\dfrac{\phi, \phi \rightarrow \psi}
{\psi}$ \vspace{5pt} \\
Conditional necessitation ($\mathtt{CN}$): & $\dfrac{\phi}
{\Fclo{\FCC} \psi \rightarrow \Fclo{\emptyset} \phi}$
\end{tabular}

\end{definition}

We let the following formulas respectively correspond to $\mathtt{S}$, $\mathtt{I}$ and $\mathtt{D}$:

\medskip

\begin{tabular}{rl}
Seriality ($\mathtt{A}\text{-}\mathtt{Ser}$): & $\Fclo{\FCC} \top$ \vspace{5pt} \\
Independence of agents ($\mathtt{A}\text{-}\mathtt{IA}$): & $(\Fclo{\FCC} \phi \land \Fclo{\FDD} \psi) \rightarrow \Fclo{\FCC \cup \FDD} (\phi \land \psi)$, where $\FCC \cap \FDD = \emptyset$ \vspace{5pt} \\
Determinism ($\mathtt{A}\text{-}\mathtt{Det}$): & $\Fclo{\FCC} (\phi \lor \psi) \rightarrow (\Fclo{\FCC} \phi \lor \Fclo{\FAG} \psi)$
\end{tabular}

\begin{definition}[Axiomatic systems for $\FXL$]
\label{definition:Axiomatic systems for XL}

For every $\FXX$ in $\FES$, the axiomatic system for $\FXL$ consists of the axioms and inference rules of $\FMCL$, and the axioms corresponding to the elements of $\FXX$.

\end{definition}

\begin{theorem}[Soundness and completeness of $\FXL$~\cite{li_completeness_2024}]

For every $\FXX$ in $\FES$, the axiomatic system for $\FXL$ given in Definition \ref{definition:Axiomatic systems for XL} is sound and complete with respect to the set of valid formulas of $\FXL$.

\end{theorem}

\subsection{Some facts about grand-coalition-first action models}

In this subsection, we present some facts about grand-coalition-first action models that will be used later.

For every set $X$ of states, we say $\Delta \subseteq \mathcal{P} (X)$ is a \Fdefs{general cover} of $X$ if $\bigcup \Delta = X$. A general cover $\Delta$ is a \Fdefs{cover} if $\emptyset \notin \Delta$. It is easy to verify the following:

\begin{fact}
\label{fact:domain}

Let $\GAM = (\FST, \FAC, \{\Fout_\FCC \mid \FCC \subseteq \FAG\}, \{\Faja_\FCC \mid \FCC \subseteq \FAG\}, \Flab)$ be a grand-coalition-first action model. 

Then, for every $\FCC \subseteq \FAG$ and $s \in \FST$, $\{ \Fout_\FCC (s, \sigma_\FCC) \mid \sigma_\FCC \in \FJA_\FCC\}$ is a general cover of $\Fout_\emptyset (s,\emptyset)$.

\end{fact}

\begin{fact}
\label{fact:one direction}

Let $\GAM = (\FST, \FAC, \{\Fout_\FCC \mid \FCC \subseteq \FAG\}, \{\Fav_\FCC \mid \FCC \subseteq \FAG\}, \Flab)$ be a grand-coalition-first action model.

Then:
\begin{enumerate}[label=(\arabic*),leftmargin=3.33em]

\item 

for all $s\in \FST$, $\FCC, \FDD \subseteq \FAG$ such that $\FCC\cap \FDD = \emptyset$, $\ja{\FCC} \in \FJA_\FCC$, and $\ja{\FDD} \in \FJA_\FDD$: $\Fout_{\FCC \cup \FDD} (s,\sigma_\FCC \cup \sigma_\FDD) \subseteq \Fout_{\FCC} (s,\sigma_\FCC) \cap \Fout_{\FDD} (s,\sigma_\FDD)$;

\item 

for all $s\in \FST$, $\FCC \subseteq \FAG$, where $\FCC = \{a_1, \dots, a_n\}$, and $\ja{\FCC} \in \FJA_\FCC$, where $\sigma_\FCC = \sigma_{a_1} \cup \dots \cup \sigma_{a_n}$: $\Fout_\FCC (s, \ja{\FCC}) \subseteq \Fout_{a_1} (s, \ja{a_1}) \cap \dots \cap \Fout_{a_n} (s, \ja{a_n})$.

\end{enumerate}

\end{fact}

Neither of the converses of the two statements holds. What follows is a counterexample for both of them:

\begin{example}
\label{example:not the intersection}
~

\begin{figure}[H]
\begin{center}
\begin{tikzpicture}
[
->=stealth,
scale=1,
every node/.style={transform shape},
]

\tikzstyle{every state}=[minimum size=10mm]

\node[state] (s-0) {$s_0$};
\node[state,position=135:{20mm} from s-0] (s-1) {$s_1$};
\node[state,position=45:{20mm} from s-0] (s-2) {$s_2$};

\node[below=5mm] (a-s-0) at (s-0) {$\{p,q\}$};
\node[above=5mm] (a-s-1) at (s-1) {$\{p\}$};
\node[above=5mm] (a-s-2) at (s-2) {$\{q\}$};

\path
(s-0) edge [] node {$(\alpha_1, \beta_1) \atop (\alpha_2, \beta_2)$} (s-1)
(s-0) edge [] node {$(\alpha_1, \beta_2) \atop (\alpha_2, \beta_1)$} (s-2)
(s-1) edge [loop left] node {$(\alpha_1, \beta_1)$} (s-1)
(s-2) edge [loop right] node {$(\alpha_1, \beta_1)$} (s-2)
;

\end{tikzpicture}

\caption{
This figure is used to show that in grand-coalition-first action models, the set of outcome states of a joint action might not contain the intersection of the sets of outcome states of the joint action's parts.
}

\label{figure:gam models outcome not intersection-compositonal}

\end{center}

\end{figure}

We suppose $\FAG = \{a,b\}$. A grand-coalition-first action model is indicated in Figure \ref{figure:gam models outcome not intersection-compositonal}, where $(\alpha_i, \beta_j)$ represents the joint action of $a$ doing $\alpha_i$ and $b$ doing $\beta_j$. 
It can be seen: (1) $\Fout_{a} (s_0, \alpha_1)$, the set of outcome states of $a$ performing $\alpha_1$ at $s_0$, is equal to $\{s_1,s_2\}$, (2) $\Fout_{b} (s_0, \beta_1)$, the set of outcome states of $b$ performing $\beta_1$ at $s_0$, is equal to $\{s_1,s_2\}$, but (3) $\Fout_\FAG (s_0, (\alpha_1, \beta_1))$, the set of outcome states of $a$ and $b$ respectively performing $\alpha_1$ and $\beta_1$ at $s_0$, does not contain the intersection of $\Fout_{a} (s_0, \alpha_1)$ and $\Fout_{b} (s_0, \beta_1)$.

\end{example}

\paragraph{Remarks}

The class of grand-coalition-first action $\mathsf{SID}$-models is $\alpha$-representable by a class of alpha neighborhood models meeting certain properties (\cite{pauly_modal_2002, goranko_strategic_2013}).
The class of grand-coalition-first action $\mathsf{SD}$-models is $\alpha$-representable by a class of alpha neighborhood models meeting certain properties (\cite{shi_representation_2024}).
Are the other six classes of grand-coalition-first action models $\alpha$-representable in a similar sense? Are the eight classes of grand-coalition-first action models $z$-representable in a similar sense? All this is yet unknown.

\newcommand{\XX}{\mathrm{X}}

\section{Single-coalition-first action models and single-coalition-first neighborhood models}
\label{section:Single-coalition-first action models and single-coalition-first neighborhood models}

\subsection{Outcome functions are compositional in many situations}
\label{section:About compositionality of availability functions and outcome functions in action models}

We say that outcome functions of a class of action models are \Fdefs{compositional} if there is a way such that for every action model in the class, the outcome function of a coalition is determined by the outcome functions of its members in this way.

Outcome functions in grand-coalition-first action models are not compositional. What follows is a counterexample:

\begin{example}
\label{example:not compositional}

We suppose $\FAG = \{a,b\}$.
Two grand-coalition-first action models $\GAM_1$ and $\GAM_2$ are depicted in Figure \ref{figure:gam models outcome not compositonal}, where $\{\Fout^1_\FCC \mid \FCC \subseteq \FAG\}$ is the set of outcome functions of $\GAM_1$ and $\{\Fout^2_\FCC \mid \FCC \subseteq \FAG\}$ is the set of outcome functions of $\GAM_2$.

In $\GAM_1$: $\Fout^1_a (s_0,\alpha_1) = \{s_1,s_2\}$, $\Fout^1_b (s_0,\beta_1) = \{s_1,s_2\}$, $\Fout^1_\FAG (s_0,(\alpha_1, \beta_1)) = \{s_1\}$.

In $\GAM_2$: $\Fout^2_a (s_0,\alpha_1) = \{s_1,s_2\}$, $\Fout^2_b (s_0,\beta_1) = \{s_1,s_2\}$, $\Fout^2_\FAG (s_0,(\alpha_1, \beta_1)) = \{s_2\}$.

Note: $\Fout^1_a (s_0,\alpha_1) = \Fout^2_a (s_0,\alpha_1)$, $\Fout^1_b (s_0,\beta_1) = \Fout^2_b (s_0,\beta_1)$, but $\Fout^1_\FAG (s_0,(\alpha_1, \beta_1)) \neq \Fout^2_\FAG$ $(s_0,(\alpha_1, \beta_1))$. This implies that outcome functions in grand-coalition-first action models are not generally compositional: there is no way such that for every grand-coalition-first action model, the outcome function of a coalition is determined by the outcome functions of its members in this way.

\begin{figure}
\begin{center}
\begin{tikzpicture}
[
->=stealth,
scale=1,
every node/.style={transform shape},
]

\tikzstyle{every state}=[minimum size=10mm]

\node[state] (s-0) {$s_0$};
\node[state,position=135:{20mm} from s-0] (s-1) {$s_1$};
\node[state,position=45:{20mm} from s-0] (s-2) {$s_2$};

\node[below=5mm] (a-s-0) at (s-0) {$\{p,q\}$};
\node[above=5mm] (a-s-1) at (s-1) {$\{p\}$};
\node[above=5mm] (a-s-2) at (s-2) {$\{q\}$};

\node[below=12mm] (m-s-0) at (s-0) {$\GAM_1$};

\path
(s-0) edge [] node {$(\alpha_1, \beta_1) \atop (\alpha_2, \beta_2)$} (s-1)
(s-0) edge [] node {$(\alpha_1, \beta_2) \atop (\alpha_2, \beta_1)$} (s-2)
(s-1) edge [loop left] node {$(\alpha_1, \beta_1)$} (s-1)
(s-2) edge [loop right] node {$(\alpha_1, \beta_1)$} (s-2)
;

\end{tikzpicture}

\vspace{10pt}

\begin{tikzpicture}
[
->=stealth,
scale=1,
every node/.style={transform shape},
]

\tikzstyle{every state}=[minimum size=10mm]

\node[state] (s-0) {$s_0$};
\node[state,position=135:{20mm} from s-0] (s-1) {$s_1$};
\node[state,position=45:{20mm} from s-0] (s-2) {$s_2$};

\node[below=5mm] (a-s-0) at (s-0) {$\{p,q\}$};
\node[above=5mm] (a-s-1) at (s-1) {$\{p\}$};
\node[above=5mm] (a-s-2) at (s-2) {$\{q\}$};

\node[below=12mm] (m-s-0) at (s-0) {$\GAM_2$};

\path
(s-0) edge [] node {$(\alpha_1, \beta_2) \atop (\alpha_2, \beta_1)$} (s-1)
(s-0) edge [] node {$(\alpha_1, \beta_1) \atop (\alpha_2, \beta_2)$} (s-2)
(s-1) edge [loop left] node {$(\alpha_1, \beta_1)$} (s-1)
(s-2) edge [loop right] node {$(\alpha_1, \beta_1)$} (s-2)
;

\end{tikzpicture}

\caption{
This figure indicates two grand-coalition-first action models for Example \ref{example:not compositional}.
}

\label{figure:gam models outcome not compositonal}

\end{center}

\end{figure}

\end{example}

There are many situations in which every agent can change a component of the state, and different components of the state that can be changed by different agents are independent of each other.
In these situations, the set of outcomes of a joint action is the intersection of the sets of outcomes of individual actions of the joint action.

What follows are two examples, and the second is from \cite{alur_alternating-time_2002}.

\begin{example}
\label{example:two}
Adam and Bob control a ship lock with two doors: a front door and a back door. Adam can make the front door closed or open, Bob can make the back door closed or open, but they cannot make both doors open at the same time.
This situation can be represented by the grand-coalition-first action model depicted in Figure \ref{figure:two}, where the set of outcomes of a joint action is the intersection of the sets of outcomes of individual actions of the joint action.

\begin{figure}
\begin{center}
\begin{tikzpicture}
[
->=stealth,
scale=1,
every node/.style={transform shape},
]

\tikzstyle{every state}=[minimum size=10mm]

\node[state] (s-1) {$s_1$};
\node[state,position=0:{30mm} from s-1] (s-2) {$s_2$};
\node[state,position=270:{30mm} from s-1] (s-3) {$s_3$};

\node[above=7mm] (a-s-1) at (s-1) {$\{c_f,c_b\}$};
\node[above=7mm] (a-s-2) at (s-2) {$\{c_f\}$};
\node[below=7mm] (a-s-3) at (s-3) {$\{c_b\}$};

\path
(s-1) edge [loop left] node {$(\mathtt{skip}, \mathtt{skip})$} (s-1)
(s-2) edge [loop right] node {$(\mathtt{skip}, \mathtt{skip})$} (s-2)
(s-3) edge [loop left] node {$(\mathtt{skip}, \mathtt{skip})$} (s-3)
(s-1) edge [above,bend left] node {$(\mathtt{skip}, \mathtt{open}\text{-}\mathtt{b})$} (s-2)
(s-2) edge [below,bend left] node {$(\mathtt{skip}, \mathtt{close}\text{-}\mathtt{b})$} (s-1)
(s-1) edge [left,bend right] node {$(\mathtt{open}\text{-}\mathtt{f}, \mathtt{skip})$} (s-3)
(s-3) edge [right,bend right] node {$(\mathtt{close}\text{-}\mathtt{f}, \mathtt{skip})$} (s-1)
;

\end{tikzpicture}

\caption{
This figure indicates a grand-coalition-first action model for the situation of Example \ref{example:two}.
Here: $c_f$ and $c_b$, respectively, express \emph{the front door is closed} and \emph{the back door is closed}; $(\mathtt{open}\text{-}\mathtt{f}, \mathtt{skip})$ represents a joint action where \emph{Adam closes the front door} and \emph{Bob does nothing}, and so on. 
It can be checked that in this model, the set of outcomes of a joint action is the intersection of the sets of outcomes of individual actions of the joint action.
For example:
(1) $\Fout_\FAG (s_1, (\mathtt{skip}, \mathtt{skip})) = \Fout_a (s_1, \mathtt{skip}) \cap \Fout_b (s_1, \mathtt{skip}) = \{s_1,s_2\} \cap \{s_1,s_3\} = \{s_1\}$; 
(2) $\Fout_\FAG (s_2, (\mathtt{skip},\mathtt{close}\text{-}\mathtt{b})) = \Fout_a (s_2, \mathtt{skip}) \cap \Fout_b (s_2, \mathtt{close}\text{-}\mathtt{b}) = \{s_1,s_2\} \cap \{s_1\} = \{s_1\}$; 
(3) $\Fout_{\{a,b\}} (s_3, (\mathtt{close}\text{-}\mathtt{f}, \mathtt{skip})) = \Fout_a (s_3, \mathtt{close}\text{-}\mathtt{f}) \cap \Fout_b (s_3, \mathtt{skip}) = \{s_1\} \cap \{s_1,s_3\} = \{s_1\}$.
}

\label{figure:two}

\end{center}

\end{figure}

\end{example}

\begin{example}
\label{example:two processes}

A system has two processes, $a$ and $b$. The process $a$ can assign $1$ or $0$ to the variable $x$, and the process $b$ can assign $1$ or $0$ to the variable $y$.
When $x = 0$, $a$ can leave the value of $x$ unchanged or change it to $1$. When $x = 1$, $a$ leaves the value of $x$ unchanged.
Similarly, when $y = 0$, $b$ can either leave the value of $y$ unchanged or change it to $1$. When $y = 1$, $b$ leaves the value of $y$ unchanged.
Figure \ref{figure:gam models outcome compositonal} indicates a grand-coalition-first model for this scenario, where the set of outcomes of a joint action is the intersection of the sets of outcomes of individual actions of the joint action.

\begin{figure}
\begin{center}
\begin{tikzpicture}
[
->=stealth,
scale=1,
every node/.style={transform shape},
]

\tikzstyle{every state}=[minimum size=10mm]

\node[state] (s-1) {$s_1$};
\node[state,position=0:{25mm} from s-1] (s-2) {$s_2$};
\node[state,position=270:{25mm} from s-1] (s-3) {$s_3$};
\node[state,position=0:{25mm} from s-3] (s-4) {$s_4$};

\node[above=7mm] (a-s-1) at (s-1) {$\{x_1,y_1\}$};
\node[above=7mm] (a-s-2) at (s-2) {$\{x_1,y_0\}$};
\node[below=7mm] (a-s-3) at (s-3) {$\{x_0,y_1\}$};
\node[below=7mm] (a-s-3) at (s-4) {$\{x_0,y_0\}$};

\path
(s-1) edge [loop left] node {$(\mathtt{skip}, \mathtt{skip})$} (s-1)
(s-2) edge [loop right] node {$(\mathtt{skip}, \mathtt{skip})$} (s-2)
(s-3) edge [loop left] node {$(\mathtt{skip}, \mathtt{skip})$} (s-3)
(s-4) edge [loop right] node {$(\mathtt{skip}, \mathtt{skip})$} (s-4)
(s-2) edge [above] node {$(\mathtt{skip}, \mathtt{y:=1})$} (s-1)
(s-3) edge [left] node {$(\mathtt{x:=1}, \mathtt{skip})$} (s-1)
(s-4) edge [] node {$(\mathtt{x:=1}, \mathtt{y:=1})$} (s-1)
(s-4) edge [right] node {$(\mathtt{x:=1}, \mathtt{skip})$} (s-2)
(s-4) edge [below] node {$(\mathtt{skip}, \mathtt{y:=1})$} (s-3)
;

\end{tikzpicture}

\caption{
This figure indicates a grand-coalition-first action model for the situation of Example \ref{example:two processes}.
It can be checked that in this model, the set of outcomes of a joint action is the intersection of the sets of outcomes of individual actions of the joint action.
For example:
(1) $\Fout_\FAG (s_1, (\mathtt{skip}, \mathtt{skip})) = \Fout_{a}(s_1, \mathtt{skip}) \cap \Fout_{b}(s_1, \mathtt{skip}) = \{s_1\} \cap \{s_1\} = \{s_1\}$;
(2) $\Fout_\FAG (s_2, (\mathtt{skip}, \mathtt{y:=1})) = \Fout_{a}(s_2, \mathtt{skip}) \cap \Fout_{b}(s_2, \mathtt{y:=1}) = \{s_1, s_2\} \cap \{s_1\} = \{s_1\}$;
(3) $\Fout_\FAG (s_4, (\mathtt{x:=1}, \mathtt{y:=1})) = \Fout_{a}(s_4, \mathtt{x:=1}) \cap \Fout_{b}(s_4, \mathtt{y:=1}) = \{s_1,s_2\} \cap \{s_1,s_3\} = \{s_1\}$.
}

\label{figure:gam models outcome compositonal}

\end{center}

\end{figure}

\end{example}

\paragraph{Remarks}

We want to mention that some literature, such as \cite{harrenstein_boolean_2001} and \cite{hoek_logic_2005}, deals with agency by \emph{propositional control}: there are some agents; every atomic proposition is controlled by at most one agent. In these settings, outcome functions are compositional.

\subsection{Single-coalition-first action models}

\begin{definition}[Single-coalition-first action models]
\label{definition:Single-coalition-first action models}

An action model $\SAM = (\FST, \FAC, \{\Fout_\FCC \mid \FCC \subseteq \FAG\}, \{\Faja_\FCC \mid \FCC \subseteq \FAG\}, \Flab)$ is a single-coalition-first action model if the following conditions are met:
\begin{itemize}

\item 

for all $a \in \FAG$ and $s \in \FST$: $\{ \Fout_a (s, \sigma_a) \mid \sigma_a \in \FJA_a \}$ is a general cover of $\Fout_\emptyset (s,\emptyset)$.

\item 

for all nonempty $\FCC \subseteq \FAG$, $s \in \FST$ and $\sigma_\FCC \in \FJA_\FCC$: $\Fout_\FCC (s, \ja{\FCC}) = \bigcap \{\Fout_a (s, \ja{\FCC}|_a) \mid a \in \FCC\}$.

\end{itemize}

\end{definition}

It can be observed that for every $\FCC \subseteq \FAG$, $\Fout_\FCC$ can be derived from $\Fout_\emptyset$ and all $\Fout_a$ where $a \in \FCC$. This is why we use ``single-coalition-first''.

Clearly, the outcome functions in single-coalition-first action models are compositional.

For all $\sigma_\FCC, \sigma'_\FCC, \sigma''_\FCC \in \FJA_\FCC$, we say that $\sigma''_\FCC$ is a \Fdefs{fusion} of $\sigma_\FCC$ and $\sigma'_\FCC$, if for all $a \in \FCC$, $\sigma''_\FCC |_a$ is equal to $\sigma_\FCC |_a$ or $\sigma'_\FCC |_a$.

The following result offers two alternative definitions of single-coalition-first action models.

\begin{theorem}
\label{theorem:equivalent condition}

Let $\FAM = (\FST, \FAC, \{\Fout_\FCC \mid \FCC \subseteq \FAG\}, \{\Faja_\FCC \mid \FCC \subseteq \FAG\}, \Flab)$ be an action model and $s \in \FST$.

The following three sets of conditions are equivalent:
\begin{enumerate}[label=(\arabic*),leftmargin=3.33em]

\item
\begin{enumerate}

\item

for all $a \in \FAG$, $\{ \Fout_a (s, \sigma_a) \mid \sigma_a \in \FJA_a \}$ is a general cover of $\Fout_\emptyset (s, \emptyset)$;

\item

for all nonempty $\FCC \subseteq \FAG$ and $\sigma_\FCC \in \FJA_\FCC$, $\Fout_\FCC (s, \ja{\FCC}) = \bigcap \{\Fout_a (s, \ja{\FCC}|_a) \mid a \in \FCC \}$.

\end{enumerate}

\item
\begin{enumerate}

\item 

for all $a \in \FAG$, $\{ \Fout_a (s, \sigma_a) \mid \sigma_a \in \FJA_a \}$ is a general cover of $\Fout_\emptyset (s, \emptyset)$;

\item

for all $\FCC,\FDD \subseteq \FAG$ such that $\FCC \cap \FDD = \emptyset$, $\sigma_\FCC \in \FJA_\FCC$, and $\sigma_\FDD \in \FJA_\FDD$: $\Fout_{\FCC \cup \FDD} (s, \sigma_\FCC \cup \sigma_\FDD) =\Fout_\FCC (s, \sigma_\FCC) \cap \Fout_\FDD (s, \sigma_\FDD)$.

\end{enumerate}

\item 
\begin{enumerate}

\item 

for all $\FCC \subseteq \FAG$ and $\sigma_\FCC \in \FJA_\FCC$, $\Fout_\FCC (s, \ja{\FCC}) = \bigcup \{\Fout_\FAG (s, \ja{\FAG}) \mid \ja{\FAG} \in \FJA_\FAG \text{ and } \ja{\FCC} \subseteq \ja{\FAG}\}$;

\item 

for all $\FCC \subseteq \FAG$ and $\sigma_\FCC, \sigma_\FCC', \sigma_\FCC'' \in \FJA_\FCC$, if $\sigma''_\FCC$ is a fusion of $\sigma_\FCC$ and $\sigma'_\FCC$, then $\Fout_\FCC (s,\sigma_\FCC) \cap \Fout_\FCC (s,\sigma'_\FCC) \subseteq \Fout_\FCC (s,\sigma''_\FCC)$.

\end{enumerate}

\end{enumerate}

\end{theorem}

\begin{proof}~

(1) $\Rightarrow$ (2)

Assume (1). We only need to show (2b). Let $\FCC, \FDD \subseteq \FAG$ such that $\FCC \cap \FDD = \emptyset$, $\sigma_\FCC \in \FJA_\FCC$, and $\sigma_\FDD \in \FJA_\FDD$. We want to show $\Fout_{\FCC \cup \FDD} (s, \sigma_\FCC \cup \sigma_\FDD) = \Fout_\FCC (s, \sigma_\FCC) \cap \Fout_\FDD (s, \sigma_\FDD)$.

Assume $\FCC = \emptyset$ and $\FDD = \emptyset$. Then, $\Fout_{\FCC \cup \FDD} (s, \sigma_\FCC \cup \sigma_\FDD) = \Fout_\emptyset (s,\emptyset) = \Fout_\emptyset (s, \emptyset) \cap \Fout_\emptyset (s,\emptyset) = \Fout_\FCC (s, \sigma_\FCC) \cap \Fout_\FDD (s, \sigma_\FDD)$.

Assume $\FCC = \emptyset$ and $\FDD \neq \emptyset$. Then $\Fout_{\FCC \cup \FDD} (s, \sigma_\FCC \cup \sigma_\FDD) = \Fout_\FDD (s, \sigma_\FDD)$ and $\Fout_\FCC (s, \sigma_\FCC) \cap \Fout_\FDD (s, \sigma_\FDD) = \Fout_\emptyset (s, \emptyset) \cap \Fout_\FDD (s, \sigma_\FDD)$. It suffices to show $\Fout_\FDD (s, \sigma_\FDD) \subseteq \Fout_\emptyset (s,\emptyset)$.
Let $t \in \Fout_\FDD (s, \sigma_\FDD)$. Let $a \in \FDD$. By (1b), $t \in \Fout_a (s, \sigma_\FDD|_a)$. Then, $t \in \bigcup \{ \Fout_a (s, \sigma_a) \mid \sigma_a \in \FJA_a \} = \Fout_\emptyset (s,\emptyset)$.
Then, $t \in \Fout_\emptyset (s,\emptyset)$.
Then, $\Fout_\FDD (s, \sigma_\FDD) \subseteq \Fout_\FCC (s,\sigma_\FCC)$.
Then, $\Fout_\FCC (s, \sigma_\FCC) \cap \Fout_\FDD (s, \sigma_\FDD)$ $ = \Fout_\FDD (s, \sigma_\FDD)$.
Then, $\Fout_{\FCC \cup \FDD} (s, \sigma_\FCC \cup \sigma_\FDD) = \Fout_\FCC (s, \sigma_\FCC) \cap \Fout_\FDD (s, \sigma_\FDD)$.

The case that $\FCC \neq \emptyset$ and $\FDD = \emptyset$ can be handled similarly.

Assume $\FCC \neq \emptyset$ and $\FDD \neq \emptyset$. 
Let $\FCC = \{a_1, \dots, a_n\}$, $\FDD = \{b_1, \dots, b_m\}$, $\sigma_\FCC = \sigma_{a_1} \cup \dots \cup \sigma_{a_n}$, and $\sigma_\FDD = \sigma_{b_1} \cup \dots \cup \sigma_{b_m}$.
By (1b), $\Fout_{\FCC \cup \FDD} (s, \sigma_\FCC \cup \sigma_\FDD) = \Fout_{a_1} (s,\sigma_{a_1}) \cap \dots \cap \Fout_{a_n} (s,\sigma_{a_n}) \cap \Fout_{b_1} (s,\sigma_{b_1}) \cap \dots \cap \Fout_{b_m} (s,\sigma_{b_m}) = \Fout_\FCC (s,\sigma_\FCC) \cap \Fout_\FDD (s,\sigma_\FDD)$.

(2) $\Rightarrow$ (1)

Assume (2). It suffices to show (1b). Let $\FCC$ be a nonempty subset of $\FAG$ and $\sigma_\FCC \in \FJA_\FCC$. Let $\FCC = \{a_1, \dots, a_n\}$ and $\sigma_\FCC = \sigma_{a_1} \cup \dots \cup \sigma_{a_n}$. 
 
Note $\FCC = \{a_1\} \cup \FCC - \{a_1\}$. By (2b), $\Fout_\FCC (s, \sigma_\FCC) = \Fout_{a_1} (s, \sigma_{a_1}) \cap \Fout_{\FCC - \{a_1\}} (s, \sigma_\FCC|_{\FCC - \{a_1\}})$.

Note $\FCC - \{a_1\} = \{a_2\} \cup \FCC - \{a_1, a_2\}$. By (2b), $\Fout_{\FCC - \{a_1\}} (s, \sigma_\FCC|_{\FCC - \{a_1\}}) = \Fout_{a_2} (s, \sigma_{a_2}) \cap \Fout_{\FCC - \{a_1, a_2\}} (s, \sigma_\FCC|_{\FCC - \{a_1, a_2\}})$.
 
\dots. 
 
Finally, we have $\Fout_\FCC (s, \sigma_\FCC) = \Fout_{a_1} (s, \sigma_{a_1}) \cap \dots \cap \Fout_{a_n} (s, \sigma_{a_n}) = \bigcap \{\Fout_a (s, \ja{a}) \mid a \in \FCC, \ja{a} \in \FJA_a \text{ and } \ja{a} \subseteq \ja{\FCC}\}$.

(2) $\Rightarrow$ (3)

Assume (2).

First, we show (3a). Let $\FCC \subseteq \FAG$ and $\sigma_\FCC \in \FJA_\FCC$. It suffices to show $\Fout_\FCC (s, \sigma_\FCC) = \bigcup \{\Fout_\FAG (s, \ja{\FAG}) \mid \ja{\FAG} \in \FJA_\FAG \text{ and } \ja{\FCC} \subseteq \ja{\FAG}\}$.

Let $w \in \Fout_\FCC (s, \sigma_\FCC)$. We want to show $w \in \bigcup \{\Fout_\FAG (s, \ja{\FAG}) \mid \ja{\FAG} \in \FJA_\FAG \text{ and } \ja{\FCC} \subseteq \ja{\FAG}\}$.
It suffices to show there is $\sigma_\FAG \in \FJA_\FAG$ such that $\sigma_\FCC \subseteq \sigma_\FAG \text{ and}$ $w \in \Fout_\FAG (s, \sigma_\FAG)$.
By (2b), it suffices to show that there is $\sigma_\FCCb \in \FJA_\FCCb$ such that $w \in \Fout_\FCCb (s, \sigma_\FCCb)$.

By (1a) and (1b), it is easy to verify $w \in \Fout_\emptyset (s,\emptyset)$.
Assume $\FCCb = \emptyset$. Clearly, what we want holds.
Assume $\FCCb \neq \emptyset$. Let $\FCCb = \{a_1, \dots, a_n\}$. 
Note for every $a\in \FAG$, $\bigcup \{ \Fout_a (s, \sigma_a)\mid\sigma_a \in \FJA_a \} = \Fout_\emptyset (s,\emptyset)$.
Then, there is $\sigma_{a_1} \in \FJA_{a_1}$ such that $w \in \Fout_{a_1} (s, \sigma_{a_1})$, \dots, there is $\sigma_{a_n} \in \FJA_{a_n}$ such that $w \in \Fout_{a_n} (s, \sigma_{a_n})$.
Let $\sigma_\FCCb = \sigma_{a_1} \cup \dots \cup \sigma_{a_n}$. By (1b), $w \in \Fout_\FCCb (s, \sigma_\FCCb)$.

Let $w \in \bigcup \{\Fout_\FAG (s, \ja{\FAG}) \mid \ja{\FAG} \in \FJA_\FAG \text{ and } \ja{\FCC} \subseteq \ja{\FAG}\}$. Then, there is $\sigma_\FAG \in \FJA_\FAG$ such that $\sigma_\FCC \subseteq \sigma_\FAG \text{ and}$ $w \in \Fout_\FAG (s, \sigma_\FAG)$. Let $\sigma_\FAG = \sigma_\FCC \cup \sigma_\FCCb$.
By (2b), $w \in \Fout_\FCC (s, \sigma_\FCC)$.

Second, we show (3b). Let $\FCC \subseteq \FAG$ and $\sigma_\FCC, \sigma_\FCC', \sigma_\FCC'' \in \FJA_\FCC$ such that $\sigma''_\FCC$ is a fusion of $\sigma_\FCC$ and $\sigma'_\FCC$. We want to show $\Fout_\FCC (s,\sigma_\FCC) \cap \Fout_\FCC (s,\sigma'_\FCC) \subseteq \Fout_\FCC (s,\sigma''_\FCC)$.
Let $t \in \Fout_\FCC (s,\sigma_\FCC) \cap \Fout_\FCC (s,\sigma'_\FCC)$.

Assume $\FCC = \emptyset$. It is easy to see $t \in \Fout_\FCC (s,\sigma''_\FCC)$.

Assume $\FCC \neq \emptyset$. Let $\FCC = \{c_1, \dots, c_n\}$.
Let $\sigma_\FCC = \sigma_{c_1} \cup \dots \cup \sigma_{c_n}$, $\sigma'_\FCC = \sigma'_{c_1} \cup \dots \cup \sigma'_{c_n}$, and $\sigma''_\FCC=\sigma''_{c_1} \cup \dots \cup \sigma''_{c_n}$, where $\sigma''_{c_i} = \sigma_{c_i}$ or $\sigma''_{c_i} = \sigma'_{c_i}$ for every $i$ such that $1 \leq i \leq n$.
By (1b), $t \in \Fout_{c_1} (s,\sigma_{c_1})$, $t \in \Fout_{c_1} (s,\sigma'_{c_1})$, \dots, $t \in \Fout_{c_n} (s,\sigma_{c_n})$,  $t \in \Fout_{c_n} (s,\sigma'_{c_n})$.
Then, $t \in \Fout_{c_1} (s,\sigma''_{c_1})$, \dots, $t \in \Fout_{c_n} (s,\sigma''_{c_n})$. By (1b), $t \in \Fout_\FCC (s,\sigma''_\FCC)$.

(3) $\Rightarrow$ (2)

Assume (3).

First, we show (2a).
Let $a \in \FAG$. We want to show $\bigcup\{ \Fout_a (s, \sigma_a) \mid \sigma_a \in \FJA_a \}= \Fout_\emptyset (s,\emptyset)$.

Let $t\in \bigcup\{ \Fout_a (s, \sigma_a) \mid \sigma_a \in \FJA_a \}$.
Then, there is $\sigma_a \in \FJA_a$ such that $t \in \Fout_a (s, \sigma_a)$. By (3a), there is $\sigma_\FAG$ such that $\sigma_a \subseteq \sigma_\FAG$ and $t \in \Fout_\FAG(s, \sigma_\FAG)$. Then, $t \in \bigcup \{\Fout_\FAG (s, \ja{\FAG}) \mid \ja{\FAG} \in \FJA_\FAG \text{ and } \emptyset \subseteq \ja{\FAG}\}$, which is equal to $\Fout_\emptyset (s,\emptyset)$ by (3a).

Let $t \in \Fout_\emptyset (s,\emptyset)$. By (3a), $t \in \bigcup \{\Fout_\FAG (s, \ja{\FAG}) \mid \ja{\FAG} \in \FJA_\FAG \text{ and } \emptyset \subseteq \ja{\FAG}\}$. Then, there is $\sigma_\FAG$ such that $t \in \Fout_\FAG (s,\sigma_\FAG)$.
By (3a), $t \in \Fout_a (s,\sigma_\FAG|_a)$.
Then, $t\in \bigcup\{ \Fout_a (s, \sigma_a) \mid \sigma_a \in \FJA_a \}$.

Second, we show (2b). 
Let $\FCC, \FDD \subseteq \FAG$ such that $\FCC \cap \FDD = \emptyset$, $\sigma_\FCC \in \FJA_\FCC$, and $\sigma_\FDD \in \FJA_\FDD$.
We want to show $\Fout_{\FCC \cup \FDD} (s, \sigma_\FCC \cup \sigma_\FDD) =\Fout_\FCC (s, \sigma_\FCC) \cap \Fout_\FDD (s, \sigma_\FDD)$.

Assume $t \in \Fout_{\FCC \cup \FDD} (s, \sigma_\FCC \cup \sigma_\FDD)$. By (3a), $t \in \Fout_\FAG (s, \sigma_\FAG)$ for some $\sigma_\FAG \in \FJA_\FAG$ such that $\sigma_\FCC \cup \sigma_\FDD \subseteq \sigma_\FAG$. Note $\sigma_\FCC \subseteq \sigma_\FAG$ and $\sigma_\FDD \subseteq \sigma_\FAG$. By (3a), $t \in \Fout_\FCC (s, \sigma_\FCC)$ and $t \in \Fout_\FDD (s, \sigma_\FDD)$. Then, $t \in \Fout_\FCC (s, \sigma_\FCC) \cap \Fout_\FDD (s, \sigma_\FDD)$.

Assume $t \in \Fout_\FCC (s, \sigma_\FCC) \cap \Fout_\FDD (s, \sigma_\FDD)$.
Then, there is $\sigma_\FAG$ and $\sigma'_\FAG$ such that $\sigma_\FCC \subseteq \sigma_\FAG$, $\sigma_\FDD \subseteq \sigma'_\FAG$, $t \in  \Fout_\FAG (s, \sigma_\FAG)$ and $t \in \Fout_\FAG (s, \sigma'_\FAG)$.
Assume $t \notin \Fout_{\FCC \cup \FDD} (s, \sigma_\FCC \cup \sigma_\FDD)$. We want to get a contradiction.

We claim $\sigma_\FCC \nsubseteq \sigma'_\FAG$ and $\sigma_\FDD \nsubseteq \sigma_\FAG$. Assume $\sigma_\FCC \subseteq \sigma'_\FAG$. Note $\FCC \cap \FDD = \emptyset$. Then $\sigma_\FCC \cup \sigma_\FDD \subseteq \sigma'_\FAG$. By (3a), $t \in \Fout_{\FCC \cup \FDD} (s, \sigma_\FCC \cup \sigma_\FDD)$, which is impossible. Similarly, we can show $\sigma_\FDD \nsubseteq \sigma_\FAG$.

Let $\sigma'_\FCC \in \FJA_\FCC$ and $\sigma'_\FDD \in \FJA_\FDD$ such that $\sigma'_\FCC \subseteq \sigma'_\FAG$ and $\sigma'_\FDD \subseteq \sigma_\FAG$. Then, $\sigma'_\FCC \cup \sigma_\FDD \subseteq \sigma'_\FAG$ and $\sigma_\FCC \cup \sigma'_\FDD \subseteq \sigma_\FAG$.
By (3a), $t \in \Fout_{\FCC \cup \FDD} (s, \sigma'_\FCC \cup \sigma_\FDD)$ and $t \in \Fout_{\FCC \cup \FDD} (s, \sigma_\FCC \cup \sigma'_\FDD)$.
Note $\sigma_\FCC \cup \sigma_\FDD$ is a fusion of $\sigma'_\FCC \cup \sigma_\FDD$ and $\sigma_\FCC \cup \sigma'_\FDD$. By (3b), $t \in \Fout_{\FCC \cup \FDD} (s, \sigma_\FCC \cup \sigma_\FDD)$.
We have a contradiction.
 
\end{proof}

\begin{theorem}
\label{theorem:implication}

Every single-coalition-first action model is a grand-coalition-first action model, but not vice versa.

\end{theorem}

The first part of the result follows from the previous theorem. For the second part, the grand-coalition-first action model given in Example \ref{example:not the intersection} is a counterexample.

\paragraph{Remarks}

Goranko and Jamroga \cite{goranko_comparing_2004} discussed \emph{convex concurrent game models}, which are closely related to \emph{rectangular game forms} in game theory~\cite{abdou_rectangularity_1998}. It can be verified that convex concurrent game models are single-coalition-first action $\mathtt{SID}$-models.

\subsection{Single-coalition-first neighborhood models}

Let $\FST$ be a nonempty set of states and $\Delta_1, \Delta_2 \subseteq \mathcal{P} (\FST) - \{\emptyset\}$. Define $\Delta_1 \odot \Delta_2 = \{\FYY_1 \cap \FYY_2 \mid \FYY_1 \in \Delta_1, \FYY_2 \in \Delta_2, \text{and } \FYY_1 \cap \FYY_2 \neq \emptyset\}$.
Let $I$ be a nonempty set of indices. Define $\bigodot \{\Delta_i \mid i \in I \text{ and } \Delta_i \subseteq \mathcal{P} (\FST) - \{\emptyset\} \}$ as expected\footnote{Strictly speaking, we should write $\bigodot \Big\{(i, \Delta_i) \mid i \in I \text{ and } \Delta_i \subseteq \mathcal{P} (\FST) - \{\emptyset\} \Big\}$. Here, we abuse the notation a bit.}.
Note $\bigodot \{\Delta\} = \Delta$.

\begin{definition}[Single-coalition-first neighborhood models]
\label{definition:??}

A neighborhood model $\SNM = (\FST,$ $\{\Fnei_\FCC \mid \FCC \subseteq \FAG\}, \Flab)$ is a \Fdefs{single-coalition-first neighborhood model} if the following conditions are met:
\begin{itemize}

\item

for all $s \in \FST$, $\Fnei_\emptyset (s)$ is either empty or a singleton with a nonempty element.

\item 

for all $a \in \FAG$ and $s \in \FST$, $\Fnei_a (s)$ is a cover of $\bigcup \Fnei_\emptyset (s)$.

\item 

for every nonempty $\FCC \subseteq \FAG$ and $s \in \FST$: $\Fnei_\FCC (s) = \bigodot \{\Fnei_a (s) \mid a \in \FCC\}$.

\end{itemize}

\end{definition}

The following fact will be used later.

\begin{fact}
\label{fact:single-coalition-first neighborhood models cover}

Let $\SNM = (\FST, \{\Fnei_\FCC \mid \FCC \subseteq \FAG\}, \Flab)$ be a single-coalition-first neighborhood model.

Then, for all $\FCC \subseteq \FAG$ and $s \in \FST$, $\Fnei_\FCC (s)$ is a cover of $\bigcup \Fnei_\emptyset (s)$.

\end{fact}

\begin{proof}~

Let $\FCC \subseteq \FAG$ and $s \in \FST$.
Assume $\bigcup \Fnei_\emptyset (s) = \emptyset$. Then $\Fnei_\FCC (s) = \emptyset$, which is a cover of $\bigcup \Fnei_\emptyset (s)$.
Assume $\bigcup \Fnei_\emptyset (s) \neq \emptyset$ and $\FCC = \emptyset$. Then $\Fnei_\FCC (s) = \{\bigcup \Fnei_\emptyset (s)\}$, which is a cover of $\bigcup \Fnei_\emptyset (s)$.
Assume $\bigcup \Fnei_\emptyset (s) \neq \emptyset$ and $\FCC \neq \emptyset$. Let $\FCC = \{a_1, \dots, a_n\}$.
Let $t \in \bigcup \Fnei_\emptyset (s)$. We want to show $t \in \bigcup \Fnei_\FCC (s)$.
Note $\Fnei_a (s)$ is a cover of $\bigcup \Fnei_\emptyset (s)$ for every $a \in \FAG$.
Then, there is $\XX_{a_1} \in \Fnei_{a_1} (s)$ such that $t \in \XX_{a_1}$, \dots, there is $\XX_{a_n} \in \Fnei_{a_n} (s)$ such that $t \in \XX_{a_n}$. Then $t \in \XX_{a_1} \cap \dots \cap \XX_{a_n}$. Note $\XX_{a_1} \cap \dots \cap \XX_{a_n} \in \Fnei_\FCC (s)$. Then $t \in \bigcup \Fnei_\FCC (s)$.

\end{proof}

\begin{definition}[Seriality, independence, and determinism of single-coalition-first neighborhood models]
\label{definition:??}

Let $\SNM = (\FST, \{\Fnei_\FCC \mid \FCC \subseteq \FAG\}, \Flab)$ be a single-coalition-first neighborhood model.

We say:
\begin{itemize}

\item

$\SNM$ is \Fdefs{serial} if for all $s \in \FST$ and $\FCC \subseteq \FAG$, $\Fnei_\FCC (s) \neq \emptyset$.

\item

$\SNM$ is \Fdefs{independent} if for all $s\in \FST$, $\FCC, \FDD \subseteq \FAG$ such that $\FCC \cap \FDD = \emptyset$, $\FYY_1 \in \Fnei_\FCC (s)$ and $\FYY_2 \in \Fnei_\FDD (s)$: $\FYY_1 \cap \FYY_2 \neq \emptyset$.

\item

$\SNM$ is \Fdefs{deterministic} if for all $s \in \FST$ and $\FYY \in \Fnei_\FAG (s)$, $\FYY$ is a singleton.

\end{itemize}

\end{definition}

As above, we let the eight strings $\epsilon$, $\mathtt{S}$, $\mathtt{I}$, $\mathtt{D}$, $\mathtt{SI}$, $\mathtt{SD}$, $\mathtt{ID}$, and $\mathtt{SID}$ in the set $\FES$ correspond to the eight combinations of the three properties, respectively.

For any $\FXX \in \FES$ and single-coalition-first neighborhood model $\SNM$, we say $\FSNM$ is an \Fdefs{$\FXX$-model} if it has the properties corresponding to $\FXX$.

\paragraph{Remarks}

Alur, Henzinger, and Kupferman \cite{alur_alternating-time_1998} proposed a semantics for $\FATL$, where models are based on \emph{alternating transition systems}. Goranko and Jamroga \cite{goranko_comparing_2004} discussed \emph{tight} alternating transition systems, which we can verify are single-coalition-first neighborhood $\mathtt{SID}$-models.

\subsection{Representation of single-coalition-first action models by single-coalition-first neighborhood models}

\begin{theorem}
\label{theorem:representation single-coalition-first action models TO single-coalition-first neighborhood models}

Every single-coalition-first action model is $z$-representable by a single-coalition-first neighborhood model.

\end{theorem}

This result is easy to show, and we skip its proof.

\begin{theorem}[]
\label{theorem:representation single-coalition-first neighborhood models TO single-coalition-first action models}

Every single-coalition-first neighborhood model $z$-represents a single-coalition-first action model.

\end{theorem}

\newcommand{\FX}{\mathrm{X}}

\begin{proof}~

Let $\SNM = (\FST, \{\Fnei_\FCC \mid \FCC \subseteq \FAG\}, \Flab)$ be a single-coalition-first neighborhood model.

Define a single-coalition-first action model $\SAM = (\FST, \FAC, \{\Fout_\FCC \mid \FCC \subseteq \FAG\}, \{\Faja_\FCC \mid \FCC \subseteq \FAG\}, \Flab)$ as follows:
\begin{itemize}

\item

$\FAC = \{\alpha_{a-s-\XX} \mid a \in \FAG, s \in \FST, \text{and } \XX \in \Fnei_a (s) \}$.

\item

\begin{itemize}

\item 

for every $s \in \FST$, $\Fout_\emptyset (s,\emptyset) = \bigcup \Fnei_\emptyset (s)$.

\item

for every $a \in \FAG$, $s \in \FST$ and $\alpha_{x-u-\XX} \in \FAC$,
\[
\Fout_a (s, \alpha_{x-u-\XX}) = 
\begin{cases}
\XX & \text{if $x = a$ and $u = s$} \\
\emptyset & \text{otherwise}
\end{cases}
\]

\item

for all nonempty $\FCC \subseteq \FAG$, $s \in \FST$ and $\sigma_\FCC \in \FJA_\FCC$: $\Fout_\FCC (s, \ja{\FCC}) = \bigcap \{\Fout_a (s, \ja{\FCC}|_a) \mid a \in \FCC\}$.

\end{itemize}

\end{itemize}

It is easy to check that for all $a \in \FAG$ and $s \in \FST$, $\{ \Fout_a (s, \sigma_a) \mid \sigma_a \in \FJA_a \}$ is a general cover of $\Fout_\emptyset (s,\emptyset)$. Then $\SAM$ is a single-coalition-first action model.

We claim that $\FSAM$ is $z$-representable by $\FSNM$.

Let $\{\FAE_\FCC \mid \FCC \subseteq \FAG\}$ be the class of actual effectivity functions of $\SAM$. Let $\FCC \subseteq \FAG$ and $s \in \FST$. It suffices to show $\FAE_\FCC (s) = \Fnei_\FCC (s)$. Note $\FAE_\FCC (s) = \{\Fout_\FCC (s, \sigma_\FCC) \mid \sigma_\FCC \in \FJA_\FCC \text{ and } \Fout_\FCC (s, \sigma_\FCC) \neq \emptyset \}$.

Assume $\bigcup \Fnei_\emptyset (s) = \emptyset$. Then, $\Fout_\emptyset (s,\emptyset)=\emptyset$. Note for all $a \in \FAG$ and $s \in \FST$, $\{ \Fout_a (s, \sigma_a) \mid \sigma_a \in \FJA_a \}$ is a general cover of $\Fout_\emptyset (s,\emptyset)$. Then, for all $a \in \FAG$, $\sigma_a \in \FJA_a$ and $s \in \FST$, $ \Fout_a (s, \sigma_a)= \emptyset$. Then, $\FAE_\FCC (s) = \emptyset=\Fnei_\FCC (s)$.

Assume $\bigcup \Fnei_\emptyset (s) \neq \emptyset$ and $\FCC = \emptyset$. Note $\FAE_\FCC (s) = \{\bigcup \Fnei_\emptyset (s)\}$ and $\Fnei_\FCC (s) = \{\bigcup \Fnei_\emptyset (s)\}$. Then, $\FAE_\FCC (s) = \Fnei_\FCC (s)$.

Assume $\bigcup \Fnei_\emptyset (s) \neq \emptyset$ and $\FCC \neq \emptyset$. Let $\FCC = \{a_1, \dots, a_n\}$.
Note $\Fnei_\FCC (s) = \bigodot \{\Fnei_a (s) \mid a \in \FCC\}$.

Let $\XX \in \FAE_\FCC (s)$. Then, there is $\sigma_\FCC \in \Fav_\FCC (s)$ such that $\XX = \Fout_\FCC (s, \sigma_\FCC)$. Let $\sigma_\FCC = \sigma_{a_1} \cup \dots \cup \sigma_{a_n}$.
Note that $\Fout_\FCC (s, \sigma_\FCC) = \Fout_{a_1} (s, \sigma_{a_1}) \cap \dots \cap \Fout_{a_n} (s, \sigma_{a_n})$ and $\XX \neq \emptyset$. 
Then, $\Fout_{a_1} (s, \sigma_{a_1}) \neq \emptyset$, \dots, $\Fout_{a_n} (s, \sigma_{a_n}) \neq \emptyset$.
Let $i$ be such that $1 \leq i \leq n$ and $\sigma_{a_i} = \alpha_{x-u-\XX_i}$. By the definition of $\Fout_{a_i} (s, \sigma_{a_i})$, $\Fout_{a_i} (s, \sigma_{a_i}) = \XX_i$, $x = a_i$, and $u = s$. Then $\sigma_{a_i} = \alpha_{a_i-s-\XX_i}$ and $\XX_i \in \Fnei_{a_i} (s)$.
Then, $\XX = \XX_1 \cap \dots \cap \XX_n$. Then, $\XX \in \Fnei_\FCC (s)$.

Let $\XX \in \Fnei_\FCC (s)$. Then, $\XX = \XX_1 \cap \dots \cap \XX_n$ for some $\XX_1 \in \Fnei_{a_1} (s), \dots, \XX_n \in \Fnei_{a_n} (s)$.
Then, $\alpha_{a_1-s-\XX_1} \dots, \alpha_{a_n-s-\XX_n} \in \FAC$.
Note $\Fout_{a_1} (s, \alpha_{a_1-s-\XX_1}) = \XX_1$, $\dots$, $\Fout_{a_n} (s, \alpha_{a_n-s-\XX_n}) = \XX_n$.
Let $\sigma_\FCC = \alpha_{a_1-s-\XX_1} \cup \dots \cup \alpha_{a_n-s-\XX_n}$. Then, $\Fout_\FCC (s, \sigma_\FCC) = \Fout_{a_1} (s, \alpha_{a_1-s-\XX_1}) \cap \dots \cap \Fout_{a_n} (s, \alpha_{a_n-s-\XX_n}) = \XX_1 \cap \dots \cap \XX_n = \XX$. Note $\XX$ is not empty. Then, $\XX \in \FAE_\FCC (s)$.

\end{proof}

\begin{theorem}
\label{theorem:X iff X}

For every single-coalition-first action model $\SAM$ and single-coalition-first neighborhood model $\SNM$, if $\SAM$ is $z$-representable by $\SNM$, then for every $\FXX \in \FES$, $\SAM$ is an $\FXX$-model if and only if $\SNM$ is an $\FXX$-model.

\end{theorem}

\begin{proof}~

Let $\SAM = (\FST, \FAC, \{\Fout_\FCC \mid \FCC \subseteq \FAG\}, \{\Fav_\FCC \mid \FCC \subseteq \FAG\}, \Flab)$ be a single-coalition-first action model, and $\{\FAE_\FCC \mid \FCC \subseteq \FAG\}$ be the class of actual effectivity functions of $\SAM$.
Let $\SNM = (\FST, \{\Fnei_\FCC \mid \FCC \subseteq \FAG\}, \Flab)$ be a single-coalition-first neighborhood model.
Assume $\SAM$ is $z$-representable by $\SNM$. Then, for all $\FCC \subseteq \FAG$, $\FAE_\FCC = \Fnei_\FCC$.

Let $\FXX \in \FES$. It is easy to see that the following three groups of statements are equivalent, respectively.
\begin{enumerate}[label=(\arabic*),leftmargin=3.33em]

\item 
\begin{itemize}

\item

$\SAM$ is serial;

\item 

for all $s \in \FST$ and $\FCC \subseteq \FAG$, $\Faja_\FCC (s) \neq \emptyset$;

\item 

for all $s \in \FST$ and $\FCC \subseteq \FAG$, $ \FAE_\FCC(s) \neq \emptyset$;

\item 

for all $s \in \FST$ and $\FCC \subseteq \FAG$, $\Fnei_\FCC(s) \neq \emptyset$;

\item 

$\SNM$ is serial.

\end{itemize}

\item 

\begin{itemize}

\item

$\SAM$ is independent;

\item

for all $s \in \FST, \FCC, \FDD \subseteq \FAG$ such that $\FCC \cap \FDD = \emptyset, \ja{\FCC} \in \FJA_\FCC$, and $\ja{\FDD} \in \FJA_\FDD$, if $\ja{\FCC} \in \Faja_\FCC (s)$ and $\ja{\FDD} \in \Faja_\FDD (s)$, then $\ja{\FCC} \cup \ja{\FDD} \in \Faja_{\FCC \cup \FDD} (s)$;

\item

for all $s \in \FST, \FCC, \FDD \subseteq \FAG$ such that $\FCC \cap \FDD = \emptyset$, and $\FYY_1, \FYY_2 \subseteq \FST$, if $\FYY_1 \in \FAE_\FCC (s)$ and $\FYY_2 \in \FAE_\FDD (s)$, then $\FYY_1 \cap \FYY_2 \neq \emptyset$;

\item

for all $s \in \FST, \FCC, \FDD \subseteq \FAG$ such that $\FCC \cap \FDD = \emptyset$, and $\FYY_1, \FYY_2 \subseteq \FST$, if $\FYY_1 \in \Fnei_\FCC (s)$ and $\FYY_2 \in \Fnei_\FDD (s)$, then $\FYY_1 \cap \FYY_2 \neq \emptyset$;

\item

$\SNM$ is independent.

\end{itemize}

\item

\begin{itemize}

\item 

$\SAM$ is deterministic;

\item 

for all $s \in \FST$ and $\ja{\FAG} \in \Faja_\FAG (s)$, $\Fout_\FAG (s, \ja{\FAG})$ is a singleton;

\item 

for all $s \in \FST$ and $\FYY \in \FAE_\FAG (s)$, $\FYY$ is a singleton;

\item 

for all $s \in \FST$ and $\FYY \in \Fnei_\FAG (s)$, $\FYY$ is a singleton;

\item 

$\SNM$ is deterministic.

\end{itemize}

\end{enumerate}

\noindent It follows that $\SAM$ is an $\FXX$-model if and only if $\SNM$ is an $\FXX$-model.

\end{proof}

The following result follows from Theorems \ref{theorem:representation single-coalition-first action models TO single-coalition-first neighborhood models}, \ref{theorem:representation single-coalition-first neighborhood models TO single-coalition-first action models}, and \ref{theorem:X iff X}.

\begin{theorem}[Representation of the class of single-coalition-first action $\FXX$-models by the class of single-coalition-first neighborhood $\FXX$-models]
\label{theorem:Representation of single-coalition-first action models by single-coalition-first neighborhood models}

For every $\FXX \in \FES$, the class of single-coalition-first action $\FXX$-models is $z$-representable by the class of single-coalition-first neighborhood $\FXX$-models.

\end{theorem}

\section{Clear grand-coalition-first action models and clear single-coalition-first neighborhood models}
\label{section:Clear grand-coalition-first action models and clear single-coalition-first neighborhood models}

\subsection{Clear grand-coalition-first action models}

\begin{definition}[Clear grand-coalition-first action models]
\label{definition:??}

Let $\FGAM = (\FST, \FAC, \{\Fout_\FCC \mid \FCC \subseteq \FAG\}, \{\Faja_\FCC \mid \FCC \subseteq \FAG\}, \Flab)$ be a grand-coalition-first action model.

We say that $\FGAM$ is \Fdefs{clear} if for all $s \in \FST$ and $\sigma_\FAG, \sigma'_\FAG \in \FJA_\FAG$, if $\sigma_\FAG \neq \sigma'_\FAG$, then $\Fout_\FAG (s,\sigma_\FAG) \cap \Fout_\FAG (s,\sigma'_\FAG) = \emptyset$.

\end{definition}

Intuitively, in clear grand-coalition-first action models, different action profiles have different outcome states. The following result indicates that in clear grand-coalition-first action models, different actions have different outcome states.

\begin{theorem}
\label{theorem:three equivalent definitions clear grand-coalition-first action models}

Let $\FGAM = (\FST, \FAC, \{\Fout_\FCC \mid \FCC \subseteq \FAG\}, \{\Faja_\FCC \mid \FCC \subseteq \FAG\}, \Flab)$ be a grand-coalition-first action model, and $s \in \FST$.

The following three conditions are equivalent:
\begin{enumerate}[label=(\arabic*),leftmargin=3.33em]

\item 

for all $a \in \FAG$ and $\sigma_a, \sigma'_a \in \FJA_a$, if $\sigma_a \neq \sigma'_a$, then $\Fout_a (s,\sigma_a) \cap \Fout_a (s,\sigma'_a) = \emptyset$;

\item 

for all $\FCC \subseteq \FAG$ and $\sigma_\FCC, \sigma'_\FCC \in \FJA_\FCC$, if $\sigma_\FCC \neq \sigma'_\FCC$, then $\Fout_\FCC (s,\sigma_\FCC) \cap \Fout_\FCC (s,\sigma'_\FCC) = \emptyset$;

\item

for all $\sigma_\FAG, \sigma'_\FAG \in \FJA_\FAG$, if $\sigma_\FAG \neq \sigma'_\FAG$, then $\Fout_\FAG (s,\sigma_\FAG) \cap \Fout_\FAG (s,\sigma'_\FAG) = \emptyset$.

\end{enumerate}

\end{theorem}

\begin{proof}~

(1) $\Rightarrow$ (2)

Assume (1). Suppose (2) does not hold. Then, there is $\FCC \subseteq \FAG$ and $\sigma_\FCC, \sigma'_\FCC \in \FJA_\FCC$ such that $\sigma_\FCC \neq \sigma'_\FCC$ and $\Fout_\FCC (s,\sigma_\FCC) \cap \Fout_\FCC (s,\sigma'_\FCC) \neq \emptyset$.
Then, there is $a \in \FCC$ and $\sigma_a, \sigma'_a \in \FJA_a$ such that $\sigma_a \subseteq \sigma_\FCC$, $\sigma'_a \subseteq \sigma'_\FCC$, and $\sigma_a \neq \sigma'_a$. 
Let $t \in \Fout_\FCC (s,\sigma_\FCC) \cap \Fout_\FCC (s,\sigma'_\FCC)$. Then, there is $\sigma_\FAG$ and $\sigma'_\FAG$ such that $\sigma_\FCC \subseteq \sigma_\FAG$, $\sigma'_\FCC \subseteq \sigma'_\FAG$, $t \in \Fout_\FAG (s,\sigma_\FAG)$, and $t \in \Fout_\FAG (s,\sigma'_\FAG)$.
Then, $\sigma_a \subseteq \sigma_\FAG$ and $\sigma'_a \subseteq \sigma'_\FAG$. Then, $t \in \Fout_a (s,\sigma_a)$ and $t \in \Fout_a (s,\sigma'_a)$. Then, $\Fout_a (s,\sigma_a) \cap \Fout_a (s,\sigma'_a) \neq \emptyset$, which is impossible.

(2) $\Rightarrow$ (3)

(3) is a special case of (2).

(3) $\Rightarrow$ (1)

Assume (3). Suppose (1) does not hold. Then, there is $a \in \FAG$ and $\sigma_a, \sigma'_a \in \FJA_a$ such that $\sigma_a \neq \sigma'_a$ and $\Fout_a (s,\sigma_a) \cap \Fout_a (s,\sigma'_a) \neq \emptyset$.
Let $t \in \Fout_a (s,\sigma_a) \cap \Fout_a (s,\sigma'_a)$.  Then, there is $\sigma_\FAG$ and $\sigma'_\FAG$ such that $\sigma_a \subseteq \sigma_\FAG$, $\sigma'_a \subseteq \sigma'_\FAG$, $t \in \Fout_\FAG (s,\sigma_\FAG)$, and $t \in \Fout_\FAG (s,\sigma'_\FAG)$. Then, $\Fout_\FAG (s,\sigma_\FAG) \cap \Fout_\FAG (s,\sigma'_\FAG) \neq \emptyset$, which is impossible.

\end{proof}

For any set $\XX$ of states, we say that $\Delta \subseteq \mathcal{P} (\XX)$ is a \Fdefs{general partition} of $\XX$ if (1) $\bigcup \Delta = \XX$, and (2) for every $\FYY, \FYY' \in \Delta$, if $\FYY \neq \FYY'$, then $\FYY \cap \FYY' = \emptyset$.
A general partition $\Delta$ is a partition if $\emptyset \notin \Delta$.

\begin{fact}

Let $\FGAM = (\FST, \FAC, \{\Fout_\FCC \mid \FCC \subseteq \FAG\}, \{\Faja_\FCC \mid \FCC \subseteq \FAG\}, \Flab)$ be a grand-coalition-first action model.

Then, if $\FGAM$ is clear, then for all $\FCC \subseteq \FAG$ and $s \in \FST$, $\{ \Fout_\FCC (s, \sigma_\FCC) \mid \sigma_\FCC \in \FJA_\FCC \}$ is a general partition of $\Fout_\emptyset (s,\emptyset)$.

\end{fact}

\begin{proof}~

Assume $\FGAM$ is clear. Let $\FCC \subseteq \FAG$ and $s \in \FST$. Note $\Fout_\emptyset (s,\emptyset) = \bigcup \{ \Fout_\FAG (s, \sigma_\FAG) \mid \sigma_\FAG \in \FJA_\FAG \}$.
It is easy to show $\bigcup \{ \Fout_\FCC (s, \sigma_\FCC) \mid \sigma_\FCC \in \FJA_\FCC \} = \Fout_\emptyset (s,\emptyset)$.
Let $\XX_1, \XX_2 \in \{ \Fout_\FCC (s, \sigma_\FCC) \mid \sigma_\FCC \in \FJA_\FCC \}$ such that $\XX_1 \neq \XX_2$. Then, $\XX_1 = \Fout_\FCC (s, \sigma_\FCC^1)$ for some $\sigma_\FCC^1 \in \FJA_\FCC$ and $\XX_2 = \Fout_\FCC (s, \sigma_\FCC^2)$ for some $\sigma_\FCC^2 \in \FJA_\FCC$. Note $\sigma^1_\FCC \neq \sigma^2_\FCC$. Then, $\Fout_\FCC (s,\sigma_\FCC) \cap \Fout_\FCC (s,\sigma'_\FCC) = \emptyset$, that is, $\XX_1 \cap \XX_2 = \emptyset$.

\end{proof}

Note that the other direction of the statement of this fact might not hold. What follows is a counterexample.

\begin{example}
\label{example:not clear single-coalition-first action model}

Suppose $\FAG = \{a\}$. Let $\FGAM = (\FST, \FAC, \{\Fout_\FCC \mid \FCC \subseteq \FAG\}, \{\Faja_\FCC \mid \FCC \subseteq \FAG\}, \Flab)$ be a grand-coalition-first action model such that:
\begin{itemize}

\item

$\FST = \{s\}$;

\item 

$\Fout_\emptyset (s,\emptyset) = \FST$;

\item 

$\FAC = \{\alpha_1, \alpha_2\}$;

\item 

$\Fout_a (s, \alpha_1) = \Fout_a (s, \alpha_2) = \FST$.

\end{itemize}

\noindent It is easy to check for all $\FCC \subseteq \FAG$, $\{ \Fout_\FCC (s, \sigma_\FCC) \mid \sigma_\FCC \in \FJA_\FCC \}$ is a general partition of $\Fout_\emptyset (s,\emptyset)$. However, $\FGAM$ is not clear.

\end{example}

\begin{theorem}
\label{theorem:clear grand-coalition-first action models are single-coalition-first action models}

Every clear grand-coalition-first action model is a single-coalition-first action model, but not vice versa.

\end{theorem}

\begin{proof}~

Let $\FGAM = (\FST, \FAC, \{\Fout_\FCC \mid \FCC \subseteq \FAG\}, \{\Faja_\FCC \mid \FCC \subseteq \FAG\}, \Flab)$ be a clear grand-coalition-first action model.

To show $\FGAM$ is a single-coalition-first action model, we need to show:
\begin{enumerate}[label=(\arabic*),leftmargin=3.33em]

\item

for all $a \in \FAG$ and $s \in \FST$, $\{ \Fout_a (s, \sigma_a) \mid \sigma_a \in \FJA_a \}$ is a general cover of $\Fout_\emptyset (s,\emptyset)$.

\item

for all nonempty $\FCC \subseteq \FAG$, $s \in \FST$, and $\sigma_\FCC \in \FJA_\FCC$,
\[
\Fout_\FCC (s, \ja{\FCC}) = \bigcap \{\Fout_a (s, \ja{a}) \mid a \in \FCC, \ja{a} \in \FJA_a, \text{and } \ja{a} \subseteq \ja{\FCC}\}
\]

\end{enumerate}

By Fact \ref{fact:domain}, (1) holds. It remains to show (2).
Let $\FCC \subseteq \FAG$, $s \in \FST$, and $\sigma_\FCC \in \FJA_\FCC$. Assume $\FCC \neq \emptyset$

Let $\FCC = \{a_1, \dots, a_n\}$ and $\sigma_\FCC = \sigma_{a_1} \cup \dots \cup \sigma_{a_n}$.
By Fact \ref{fact:one direction}, $\Fout_\FCC (s, \sigma_\FCC) \subseteq \Fout_{a_1} (s, \sigma_{a_1}) \cap \dots \cap \Fout_{a_n} (s, \sigma_{a_n})$.
Let $t \in \Fout_{a_1} (s, \sigma_{a_1}) \cap \dots \cap \Fout_{a_n} (s, \sigma_{a_n})$. Then, there is $\sigma^1_\FAG$ such that $\sigma_{a_1} \subseteq \sigma^1_\FAG$ and $t \in \Fout_\FAG (s,\sigma^1_\FAG)$, \dots, there is $\sigma^n_\FAG$ such that $\sigma_{a_n} \subseteq \sigma^n_\FAG$ and $t \in \Fout_\FAG (s,\sigma^n_\FAG)$. As $\FGAM$ is clear, $\sigma^1_\FAG = \dots = \sigma^n_\FAG$. Then $\sigma_\FCC \subseteq \sigma^1_\FAG$. Then $t \in \Fout_\FCC (s, \sigma_\FCC)$.

It is easy to see that the action model given in Example \ref{example:not clear single-coalition-first action model} is a single-coalition-first action model that is not clear.

\end{proof}

\paragraph{Remarks}

Goranko and Jamroga \cite{goranko_comparing_2004} and \r{A}gotnes and Alechina \cite{agotnes_embedding_2015} discussed \emph{injective concurrent game models}, which are clear grand-coalition-first action $\mathtt{SID}$-models.
As observed by \r{A}gotnes and Alechina \cite{agotnes_embedding_2015}, the assumption behind injective concurrent game models, \emph{different available action profiles have different outcomes}, is very common in game theory.

The reason that Goranko and Jamroga \cite{goranko_comparing_2004} and \r{A}gotnes and Alechina \cite{agotnes_embedding_2015} used ``injective'' is as follows: in injective concurrent game models, the outcome function $\Fout_\FAG$ for $\FAG$ is an injective function with the set of possible states as its range. However, in grand-coalition-first models, the outcome function $\Fout_\FAG$ for $\FAG$ takes the power set of the set of possible states as its range. A grand-coalition-first action model where the outcome function $\Fout_\FAG$ for $\FAG$ is injective does not imply that it is clear. This is why we do not use ``injective'' for clear grand-coalition-first action models.

\subsection{Clear single-coalition-first neighborhood models}

\begin{definition}[Clear single-coalition-first neighborhood models]
\label{definition:??}

Let $\SNM = (\FST, \{\Fnei_\FCC \mid \FCC \subseteq \FAG\}, \Flab)$ be a single-coalition-first neighborhood model.

We say that $\SNM$ is \Fdefs{clear} if for all $a \in \FAG$ and $s \in \FST$, $\Fnei_a (s)$ is a partition of $\bigcup \Fnei_\emptyset (s)$.

\end{definition}

The following result gives an alternative definition of clear single-coalition-first neighborhood models.

\begin{theorem}
\label{theorem:three equivalent definitions clear single-coalition-first action models}

Let $\SNM = (\FST, \{\Fnei_\FCC \mid \FCC \subseteq \FAG\}, \Flab)$ be a single-coalition-first neighborhood model.

The following two conditions are equivalent:
\begin{enumerate}[label=(\arabic*),leftmargin=3.33em]

\item 

for all $a \in \FAG$ and $s \in \FST$, $\Fnei_a (s)$ is a partition of $\bigcup \Fnei_\emptyset (s)$;

\item 

for all $\FCC \subseteq \FAG$ and $s \in \FST$, $\Fnei_\FCC (s)$ is a partition of $\bigcup \Fnei_\emptyset (s)$.

\end{enumerate}

\end{theorem}

\begin{proof}~

(1) $\Rightarrow$ (2)

Assume (1). Let $\FCC \subseteq \FAG$ and $s \in \FST$.
By Fact \ref{fact:single-coalition-first neighborhood models cover}, $\Fnei_\FCC (s)$ is a cover of $\bigcup \Fnei_\emptyset (s)$. Assume $\Fnei_\FCC (s)$ is not a partition of $\bigcup \Fnei_\emptyset (s)$. Then, there is $\FYY_1, \FYY_2 \in \Fnei_\FCC (s)$ such that $\FYY_1 \neq \FYY_2$ and $\FYY_1 \cap \FYY_2 \neq \emptyset$. Without loss of any generality, we can assume there is $w$ such that $w \in \FYY_1$ but $w \notin \FYY_2$. Let $u \in \FYY_1 \cap \FYY_2$.
It is easy to see $\bigcup \Fnei_\emptyset (s) \neq \emptyset$. Then, $\Fnei_\emptyset (s) = \{\bigcup \Fnei_\emptyset (s)\}$, which is a singleton. Note $\Fnei_\FCC (s)$ is not a singleton. Then, $\FCC \neq \emptyset$. Let $\FCC = \{a_1, \dots, a_n\}$.
Then, there is $\XX_{a_1} \in \Fnei_{a_1} (s), \dots, \XX_{a_n} \in \Fnei_{a_n} (s)$ such that $\FYY_1 = \XX_{a_1} \cap \dots \cap \XX_{a_n}$, and there is $\XX'_{a_1} \in \Fnei_{a_1} (s), \dots, \XX'_{a_n} \in \Fnei_{a_n} (s)$ such that $\FYY_2 = \XX'_{a_1} \cap \dots \cap \XX'_{a_n}$.
Note $w \notin \FYY_2$. Then $w \notin \XX'_{a_i}$ for some $a_i$. Note $w \in \XX_{a_i}$. Then $\XX_{a_i} \neq \XX'_{a_i}$. Note $u \in \XX_{a_i}$ and $u \in \XX'_{a_i}$. Then $\XX_{a_i} \cap \XX'_{a_i} \neq \emptyset$. We have a contradiction.

(2) $\Rightarrow$ (1)

(1) is a special case of (2).

\end{proof}

By Theorem \ref{theorem:three equivalent definitions clear grand-coalition-first action models}, clear grand-coalition-first action models can be equivalently defined in three different ways, which respectively concern all agents, all coalitions, and the grand coalition. It might be expected that the following condition, which concerns the grand coalition, is equivalent to the two conditions given in Theorem \ref{theorem:three equivalent definitions clear single-coalition-first action models}: \emph{for all $s \in \FST$, $\Fnei_\FAG (s)$ is a partition of $\bigcup \Fnei_\emptyset (s)$}. Actually, this is not the case. What follows is a counterexample.

\begin{example}
\label{example: AG partition not equivalent to a partition}

Assume $\FAG = \{a, b\}$. Let $\SNM = (\FST, \{\Fnei_\FCC \mid \FCC \subseteq \FAG\}, \Flab)$ be a single-coalition-first neighborhood model such that:
\begin{itemize}

\item

$\FST = \{s_0, s_1, s_2, s_3\}$;

\item

$\Fnei_\emptyset (s_0) = \{ \{s_1, s_2, s_3 \} \}$;

\item 

$\Fnei_a (s_0) = \{\{s_1,s_2\},\{s_2,s_3\}\}$ and $\Fnei_b(s_0)=\{\{s_2\},\{s_1,s_3\}\}$.

\end{itemize}

Note $\Fnei_a (s_0)$ is not a partition of $\bigcup \Fnei_\emptyset (s_0)$. However, it can be checked that $\Fnei_\FAG(s_0)=\{\{s_1\},\{s_2\},\{s_3\}\}$, which is a partition of $\bigcup \Fnei_\emptyset (s_0)$.

\end{example}

\paragraph{Remarks}

Alur, Henzinger, and Kupferman \cite{alur_alternating-time_1998} proposed models for $\FATL$ based on so-called \emph{lock-step synchronous alternating transition systems}. It can be shown that they are clear single-coalition-first neighborhood $\mathtt{SID}$-models.

\subsection{Representation of clear grand-coalition-first action models by clear single-coalition-first neighborhood models}

\begin{theorem}[]
\label{theorem:representation clear grand-coalition-first action models TO clear single-coalition-first neighborhood models}

Every clear grand-coalition-first action model is $z$-representable by a clear single-coalition-first neighborhood model.

\end{theorem}

\begin{proof}~

Let $\FGAM$ be a clear grand-coalition-first action model. By Theorem \ref{theorem:clear grand-coalition-first action models are single-coalition-first action models}, $\FGAM$ is a single-coalition-first action model. It is easy to show that $\FGAM$ is $z$-representable by a clear single-coalition-first neighborhood model.

\end{proof}

\begin{theorem}[]
\label{theorem:representation clear single-coalition-first neighborhood models TO clear single-coalition-first action models}

Every clear single-coalition-first neighborhood model $z$-represents a clear grand-coalition-first action model.

\end{theorem}

\begin{proof}~

Let $\SNM = (\FST, \{\Fnei_\FCC \mid \FCC \subseteq \FAG\}, \Flab)$ be a clear single-coalition-first neighborhood model.
Define a single-coalition-first action model $\SAM = (\FST, \FAC, \{\Fout_\FCC \mid \FCC \subseteq \FAG\}, \{\Faja_\FCC \mid \FCC \subseteq \FAG\}, \Flab)$ as in the proof for Theorem \ref{theorem:representation single-coalition-first neighborhood models TO single-coalition-first action models}. As shown in that proof, $\SNM$ $z$-represents $\SAM$.
It is easy to check that $\SAM$ is clear. By Theorem \ref{theorem:implication}, it is a grand-coalition-first action model.

\end{proof}

\begin{theorem}[Representation of the class of clear grand-coalition-first action $\FXX$-models by the class of clear single-coalition-first neighborhood $\FXX$-models]
\label{theorem:Representation of clear grand-coalition-first action models by clear single-coalition-first neighborhood models}

For every $\FXX \in \FES$, the class of clear grand-coalition-first action $\FXX$-models is $z$-representable by the class of clear single-coalition-first neighborhood $\FXX$-models.

\end{theorem}

\begin{proof}~

Let $\FXX \in \FES$.
Let $\GAM$ be a clear grand-coalition-first action $\FXX$-model.
By Theorem \ref{theorem:representation clear grand-coalition-first action models TO clear single-coalition-first neighborhood models}, $\GAM$ is $z$-representable by a clear single-coalition-first neighborhood model $\SNM$.
By Theorem \ref{theorem:clear grand-coalition-first action models are single-coalition-first action models}, $\FGAM$ is a single-coalition-first action model.
By Theorem \ref{theorem:X iff X}, $\SNM$ is a clear single-coalition-first neighborhood $\FXX$-model.
Let $\SNM$ be a clear single-coalition-first neighborhood $\FXX$-model. By Theorem \ref{theorem:representation clear single-coalition-first neighborhood models TO clear single-coalition-first action models}, $\SNM$ $z$-represents a clear grand-coalition-first action model $\GAM$.
By Theorem \ref{theorem:clear grand-coalition-first action models are single-coalition-first action models}, $\FGAM$ is a single-coalition-first action model.
By Theorem \ref{theorem:X iff X}, $\SAM$ is a clear grand-coalition-first action $\FXX$-model.

\end{proof}

\section{Tree-like grand-coalition-first action models and tree-like single-coalition-first neighborhood models}
\label{section:Tree-like grand-coalition-first action models and tree-like single-coalition-first neighborhood models}

\newcommand{\Fhist}{\theta}

\subsection{Tree-like grand-coalition-first action models}

\begin{definition}[Histories in grand-coalition-first action models]
\label{definition:??}

Let $\FGAM = (\FST, \FAC, \{\Fout_\FCC \mid \FCC \subseteq \FAG\}, \{\Faja_\FCC \mid \FCC \subseteq \FAG\}, \Flab)$ be a grand-coalition-first action model.

For every $\FCC \subseteq \FAG$, a finite sequence $\Fhist_\FCC = (s_0, \sigma_\FCC^1, s_1, \dots, \sigma_\FCC^n, s_n)$, where every $s_i$ is a state and every $\sigma_\FCC^i$ is a joint action of $\FCC$, is called a \Fdefs{$\FCC$-history} from $s_0$ to $s_n$ if for every $i$ such that $0 \leq i < n$, $\sigma_\FCC^{i+1} \in \Fav_\FCC (s_i)$ and $s_{i+1} \in \Fout_\FCC (s_i, \sigma_\FCC^{i+1})$.

Specially, for every $\FCC \subseteq \FAG$ and $s \in \FST$, $s$ is called a \Fdefs{$\FCC$-history} from $s$ to $s$.

\end{definition}

\begin{definition}[Tree-like grand-coalition-first action models]
\label{definition:??}

Let $\FGAM = (\FST, \FAC, \{\Fout_\FCC \mid \FCC \subseteq \FAG\}, \{\Faja_\FCC \mid \FCC \subseteq \FAG\}, \Flab)$ be a grand-coalition-first action model.

We say that $\FGAM$ is a \Fdefs{tree-like model} if there is $r \in \FST$, called a \Fdefs{root}, such that for every $s \in \FST$, there is a unique $\FAG$-history from $r$ to $s$.

\end{definition}

Note that every tree-like model has a unique root.

The following result gives two alternative definitions of tree-like grand-coalition-first action models.

\begin{theorem}
\label{theorem:three equivalent definitions tree-like action models}

Let $\FGAM = (\FST, \FAC, \{\Fout_\FCC \mid \FCC \subseteq \FAG\}, \{\Faja_\FCC \mid \FCC \subseteq \FAG\}, \Flab)$ be a grand-coalition-first action model.

The following three conditions are equivalent:
\begin{enumerate}[label=(\arabic*),leftmargin=3.33em]

\item 

there is $r \in \FST$ such that for every $a \in \FAG$ and $s \in \FST$, there is a unique $a$-history from $r$ to $s$;

\item 

there is $r \in \FST$ such that for every $\FCC \subseteq \FAG$ and $s \in \FST$, there is a unique $\FCC$-history from $r$ to $s$;

\item

there is $r \in \FST$ such that for every $s \in \FST$, there is a unique $\FAG$-history from $r$ to $s$.

\end{enumerate}

\end{theorem}

\begin{proof}~

(1) $\Rightarrow$ (3)

Assume (1). Let $r \in \FST$ such that for every $a \in \FAG$ and $s \in \FST$, there is a unique $a$-history from $r$ to $s$. Let $s \in \FST$. We want to show that there is a unique $\FAG$-history from $r$ to $s$. Let $a \in \FAG$.

Assume $s=r$.
Note $r$ is an $\FAG$-history from $r$ to $r$.

Assume there is an $\FAG$-history $\theta_\FAG$ from $r$ to $r$ which is different from $r$. It is impossible that $\Fhist_\FAG = x$ for some $x \in \FST$.
Assume $\Fhist_\FAG = (s_0, \sigma_\FAG^1, s_1, \dots, \sigma_\FAG^n, s_n)$. It is easy to check that $\Fhist_a = (s_0, \sigma_\FAG^1|_a, s_1, \dots, \sigma_\FAG^n|_a, s_n)$ is an $a$-history from $r$ to $r$. Clearly, $\Fhist_a$ is different from $r$. Note $r$ is an $a$-history from $r$ to $r$. We have a contradiction.

Assume $s\neq r$.

Let $\Fhist_a$ be the unique $a$-history from $r$ to $s$. Note that it is impossible that $\Fhist_a = x$ for some $x \in \FST$.
Assume $\Fhist_a = (s_0, \sigma_a^1, s_1, \dots, \sigma_a^n, s_n)$.

For every $i$ such that $1 \leq i \leq n$, let $\sigma^i_\FAG$ be an action profile such that $\sigma^i_a \subseteq \sigma^i_\FAG$ and $s_i \in \Fout_\FAG (s_{i-1},\sigma^i_\FAG)$. Then, $(s_0, \sigma^1_\FAG, s_1, \dots, \sigma_\FAG^n, s_n)$ is an $\FAG$-history from $r$ to $s$.

Assume there is a different $\FAG$-history $\Fhist'_\FAG$ from $r$ to $s$. Note that it is impossible that $\Fhist'_\FAG = x$ for some $x \in \FST$. Let $\theta'_\FAG = (t_0, \lambda^1_\FAG, t_1, \dots, \lambda_\FAG^m, t_m)$.

We want to show $(s_0, \sigma^1_\FAG, s_1, \dots, \sigma_\FAG^n, s_n) = (t_0, \lambda^1_\FAG, t_1, \dots, \lambda_\FAG^m, t_m)$.

Note both $(s_0, \sigma^1_\FAG|_a, s_1, \dots, \sigma_\FAG^n|_a, s_n)$ and $(t_0, \lambda^1_\FAG|_a, t_1, \dots, \lambda_\FAG^m|_a, t_m)$ are $a$-histories from $r$ to $s$. Then $n = m$ and $s_i = t_i$ for all $i$.

Assume $\sigma^i_\FAG \neq \lambda^i_\FAG$ for some $1 \leq i \leq n$. Then, there is $b \in \FAG$ such that $\sigma^i_\FAG|_b \neq \lambda^i_\FAG|_b$.
Note both $(s_0, \sigma^1_\FAG|_b, s_1, \dots, \sigma_\FAG^n|_b, s_n)$ and $(t_0, \lambda^1_\FAG|_b, t_1, \dots, \lambda_\FAG^m|_b, t_m)$ are $b$-histories from $r$ to $s$, which are different. We have a contradiction.

(3) $\Rightarrow$ (2)

Assume (3). Let $r \in \FST$ such that for every $s \in \FST$, there is a unique $\FAG$-history from $r$ to $s$.
Let $s \in \FST$ and $\FCC \subseteq \FAG$. We want to show that there is a unique $\FCC$-history from $r$ to $s$.

Assume $s=r$. 
Note $r$ is a $\FCC$-history from $r$ to $r$.

Assume there is a $\FCC$-history $\Fhist_\FCC$ from $r$ to $r$ which is different from $r$. Note that it is impossible that $\Fhist_\FCC = x$ for some $x \in \FST$. Assume $\Fhist_\FCC = (s_0, \sigma_\FCC^1, s_1, \dots, \sigma_\FCC^n,s_n)$.

For every $1 \leq i \leq n$, let $\lambda^i_\FAG$ be an action profile such that $\sigma^i_\FCC \subseteq \lambda^i_\FAG$ and $s_i \in \Fout_\FAG (s_{i-1},\lambda^i_\FAG)$. Then, $(s_0, \lambda^1_\FAG, s_1, \dots, \lambda_\FAG^n, s_n)$ is an $\FAG$-history from $r$ to $r$, which is different from $r$.
Note $r$ is an $\FAG$-history from $r$ to $r$. We have a contradiction.

Assume $s\neq r$.

Let $\theta_\FAG$ be an $\FAG$-history from $r$ to $s$. Note that it is impossible that $\theta_\FAG = x$ for some $x \in \FST$.
Let $\theta_\FAG = (s_0, \sigma_\FAG^1, s_1, \dots, \sigma_\FAG^n, s_n)$.
It is easy to see $(s_0, \sigma_\FAG^1|_\FCC, s_1, \dots, \sigma_\FAG^n|_\FCC, s_n)$ is a $\FCC$-history from $r$ to $s$.

Assume there is a different $\FCC$-history $\theta'_\FCC$ from $r$ to $s$. Again, it is impossible that $\theta_\FCC' = x$ for some $x \in \FST$. Let $\theta'_\FCC = (t_0, \lambda^1_\FCC, t_1, \dots, \lambda^m_\FCC, t_m)$.

For every $1 \leq i \leq m$, let $\lambda^i_\FAG$ be an action profile such that $\lambda^i_\FCC \subseteq \lambda^i_\FAG$ and $t_i \in \Fout_\FAG (t_{i-1},\lambda^i_\FAG)$. Then, $(t_0, \lambda^1_\FAG, t_1, \dots, \lambda_\FAG^m, t_m)$ is an $\FAG$-history from $r$ to $s$.

It is easy to check that $(t_0, \lambda^1_\FAG, t_1, \dots, \lambda_\FAG^m, t_m)$ is different from $(s_0, \sigma_\FAG^1, s_1, \dots, \sigma_\FAG^n, s_n)$. We have a contradiction.

(2) $\Rightarrow$ (1)

(1) is a special case of (2).

\end{proof}

\begin{theorem}
\label{theorem:tree-like implies clear}

Every tree-like grand-coalition-first action model is a clear grand-coalition-first action model, but not vice versa.

\end{theorem}

\begin{proof}~

Let $\FGAM = (\FST, \FAC, \{\Fout_\FCC \mid \FCC \subseteq \FAG\}, \{\Faja_\FCC \mid \FCC \subseteq \FAG\}, \Flab)$ be a tree-like grand-coalition-first action model and $r$ be its root. Let $s \in \FST$ and $\sigma_\FAG, \sigma'_\FAG \in \FJA_\FAG$ such that $\sigma_\FAG \neq \sigma'_\FAG$. It suffices to show $\Fout_\FAG (s,\sigma_\FAG) \cap \Fout_\FAG (s,\sigma'_\FAG) = \emptyset$.
Assume $\Fout_\FAG (s,\sigma_\FAG) \cap \Fout_\FAG (s,\sigma'_\FAG) \neq \emptyset$. Let $t \in \Fout_\FAG (s,\sigma_\FAG) \cap \Fout_\FAG (s,\sigma'_\FAG)$.
Let $\theta_\FAG$ be an $\FAG$-history from $r$ to $s$. It is easy to check that both $(\theta_\FAG, \sigma_\FAG, t)$ and $(\theta_\FAG, \sigma'_\FAG, t)$ are $\FAG$-histories from $r$ to $t$, which are different. We have a contradiction.

It is easy to find a clear grand-coalition-first action model with a loop, which is not tree-like.

\end{proof}

\paragraph{Remarks}

\r{A}gotnes, Goranko, and Jamroga \cite{agotnes_alternating-time_2007} defined \emph{tree-like concurrent game models}, which are different from tree-like grand-coalition-first action $\mathtt{SID}$-models defined in this section. For the former, it is only required that for every state, there is a unique \emph{computation}, which is a sequence of states, from the root to the state.

\newcommand{\ZZ}{\mathrm{Z}}
\newcommand{\UU}{\mathrm{U}}
\newcommand{\VV}{\mathrm{V}}

\subsection{Tree-like single-coalition-first neighborhood models}

\begin{definition}[Histories in single-coalition-first neighborhood models]
\label{definition:??}

Let $\SNM = (\FST, \{\Fnei_\FCC \mid \FCC \subseteq \FAG\}, \Flab)$ be a single-coalition-first neighborhood model.

For every $\FCC \subseteq \FAG$, a finite sequence $\Fhist_\FCC = (s_0, \FYY_1, s_1, \dots, \FYY_n, s_n)$, where every $s_i$ is a state and every $\FYY_i$ is a set of states, is called a \Fdefs{$\FCC$-history} from $s_0$ to $s_n$ if for every $i$ such that $0 \leq i < n$, $\FYY_{i+1} \in \Fnei_\FCC (s_i)$ and $s_{i+1} \in \FYY_{i+1}$.

Specifically, for every $\FCC \subseteq \FAG$ and $s \in \FST$, $s$ is called a \Fdefs{$\FCC$-history} from $s$ to $s$.

\end{definition}

\begin{definition}[Tree-like single-coalition-first neighborhood models]
\label{definition:??}

Let $\SNM = (\FST, \{\Fnei_\FCC \mid \FCC \subseteq \FAG\}, \Flab)$ be a single-coalition-first neighborhood model.

We say that $\SNM$ is a \Fdefs{tree-like model} if there is $r \in \FST$, called a \Fdefs{root}, such that for every $a \in \FAG$ and $s \in \FST$, there is a unique $a$-history from $r$ to $s$.

\end{definition}

The following result is easy to verify, and we skip its proof.

\begin{fact}

Every tree-like single-coalition-first neighborhood model is a clear single-coalition-first neighborhood model, but not vice versa.

\end{fact}

The following result offers an alternative definition of tree-like single-coalition-first neighborhood models.

\begin{theorem}
\label{theorem:two equivalent definitions tree-like single-coalition-first neighborhood models}

Let $\SNM = (\FST, \{\Fnei_\FCC \mid \FCC \subseteq \FAG\}, \Flab)$ be a single-coalition-first neighborhood model.

The following two conditions are equivalent:
\begin{enumerate}[label=(\arabic*),leftmargin=3.33em]

\item 

there is $r \in \FST$ such that for every $a \in \FAG$ and $s \in \FST$, there is a unique $a$-history from $r$ to $s$;

\item 

there is $r \in \FST$ such that for every $\FCC \subseteq \FAG$ and $s \in \FST$, there is a unique $\FCC$-history from $r$ to $s$.

\end{enumerate}

\end{theorem}

\begin{proof}~

(1) $\Rightarrow$ (2)

Assume (1). Let $r \in \FST$ such that for every $a \in \FAG$ and every $s \in \FST$, there is a unique $a$-history from $r$ to $s$.
Let $s \in \FST$ and $\FCC \subseteq \FAG$. We want to show that there is a unique $\FCC$-history from $r$ to $s$.

Assume $s=r$. 
Note $r$ is a $\FCC$-history from $r$ to $r$.
Assume there is a different $\FCC$-history $\theta_\FCC$ from $r$ to $r$. Note that it is impossible that $\theta_\FCC = x$ for some $x \in \FST$. Let $\theta_\FCC = (s_0, \FYY_1, s_1, \dots, \FYY_n, s_n)$.

Assume $\FCC = \emptyset$. Then for every $i$ such that $1 \leq i \leq n$, $\FYY_i = \bigcup \Fnei_\emptyset (s_{i-1})$.
Let $a \in \FAG$. For every $1 \leq i \leq n$, let $\ZZ_i \in \Fnei_a (s_{i-1})$ such that $\ZZ_i \subseteq \FYY_i$ and $s_i \in \ZZ_i$. Then, $(s_0, \ZZ_1, s_1, \dots, \ZZ_n, s_n)$ is an $a$-history from $r$ to $r$, which is different from $r$.
Note $r$ is an $a$-history from $r$ to $r$. We have a contradiction.

Assume $\FCC \neq \emptyset$. Let $a \in \FCC$. Note $\Fnei_\FCC (s) = \bigodot \{\Fnei_a (s) \mid a \in \FCC\}$ for every $s \in \FST$. For every $1 \leq i \leq n$, let $\ZZ_i \in \Fnei_a (s_{i-1})$ such that $\FYY_i \subseteq \ZZ_i$. Then, $(s_0, \ZZ_1, s_1, \dots, \ZZ_n, s_n)$ is an $a$-history from $r$ to $r$, which is different from $r$.
Note $r$ is an $a$-history from $r$ to $r$. We have a contradiction.

Assume $s\neq r$.

Assume $\FCC = \emptyset$. Let $a \in \FAG$ and $\theta_a$ be the unique $a$-history from $r$ to $s$. It is impossible that $\theta_a = x$ for some $x \in \FST$. Let $\theta_a = (s_0, \FYY_1, s_1, \dots, \FYY_n, s_n)$.
It is easy to see $(s_0, \bigcup \Fnei_\emptyset (s_0), s_1, \dots, \bigcup \Fnei_\emptyset (s_{n-1}),$ $ s_n)$ is a $\FCC$-history from $r$ to $s$.
Let $\theta_\FCC$ be a $\FCC$-history from $r$ to $s$. We want to show $\theta_\FCC = (s_0, \bigcup \Fnei_\emptyset (s_0), s_1, \dots, \bigcup \Fnei_\emptyset (s_{n-1}), s_n)$.

It is impossible that $\theta_\FCC = x$ for some $x \in \FST$. Let $\theta_\FCC = (t_0, \bigcup \Fnei_\emptyset (t_0), t_1, \dots, \bigcup \Fnei_\emptyset (t_{m-1}),$ $t_m)$.
For every $1 \leq i \leq m$, let $\ZZ_i \in \Fnei_a (t_{i-1})$ such that $\ZZ_i \subseteq \bigcup \Fnei_\emptyset (t_{i-1})$ and $t_i \in \ZZ_i$. Then, $(t_0, \ZZ_1, t_1, \dots, \ZZ_m, t_m)$ is an $a$-history from $r$ to $s$.
Then $n = m$ and $s_i = t_i$ for all $0 \leq i \leq n$.
Then, $\theta_\FCC = (s_0, \bigcup \Fnei_\emptyset (s_0), s_1, \dots, \bigcup \Fnei_\emptyset (s_{n-1}), s_n)$.

Assume $\FCC \neq \emptyset$. Let $a \in \FCC$ and $\theta_a$ be the unique $a$-history from $r$ to $s$. Note that it is impossible that $\theta_a = x$ for some $x \in \FST$. Let $\theta_a = (s_0, \FYY_1, s_1, \dots, \FYY_n, s_n)$.
For every $1 \leq i \leq n$, let $\ZZ_i \in \Fnei_\FCC (s_{i-1})$ such that $\ZZ_i \subseteq \FYY_i$ and $s_i \in \ZZ_i$. Then, $(s_0, \ZZ_1, s_1, \dots, \ZZ_n, s_n)$ is a $\FCC$-history from $r$ to $s$.

Let $\theta_\FCC$ be a $\FCC$-history from $r$ to $s$. We want to show $\theta_\FCC = (s_0, \ZZ_1, s_1, \dots, \ZZ_n, s_n)$.

It is impossible that $\theta_\FCC = x$ for some $x \in \FST$. Let $\theta_\FCC = (t_0, \UU_1, t_1, \dots, \UU_m, t_m)$.
For every $1 \leq i \leq m$, let $\VV_i \in \Fnei_a (t_{i-1})$ such that $\UU_i \subseteq \VV_i$. Then, $(t_0, \VV_1, t_1, \dots, \VV_m, t_m)$ is an $a$-history from $r$ to $s$.
Then $n = m$ and $s_i = t_i$ for every $0 \leq i \leq n$.
Assume $\ZZ_i \neq \UU_i$ for some $1 \leq i \leq n$.
Note $\SNM$ is a clear single-coalition-first neighborhood model. Then, $\ZZ_i \cap \UU_i = \emptyset$. 
Then, $s_i \neq t_i$, which is impossible. Then, $\theta_\FCC = (s_0, \ZZ_1, s_1, \dots, \ZZ_n, s_n)$.

(2) $\Rightarrow$ (1)

(1) is a special case of (2).

\end{proof}

By Theorem \ref{theorem:three equivalent definitions tree-like action models}, tree-like grand-coalition-first action models can be equivalently defined in three different ways, which respectively concern all agents, all coalitions, and the grand coalition. It might be expected that the following condition is equivalent to the two conditions in Theorem \ref{theorem:two equivalent definitions tree-like single-coalition-first neighborhood models}: \emph{there is $r \in \FST$ such that for every $s \in \FST$, there is a unique $\FAG$-history from $r$ to $s$.}
In fact, this is not the case. Example \ref{example: AG partition not equivalent to a partition} also works here.

\subsection{Representation of tree-like grand-coalition-first action models by tree-like single-coalition-first neighborhood models}

\begin{theorem}
\label{theorem:representation tree-like grand-coalition-first action models TO tree-like single-coalition-first neighborhood models}

Every tree-like grand-coalition-first action model is $z$-representable by a tree-like single-coalition-first neighborhood model.

\end{theorem}

\begin{proof}~

Let $\FGAM$ be a tree-like grand-coalition-first action model. By Theorems \ref{theorem:tree-like implies clear} and \ref{theorem:clear grand-coalition-first action models are single-coalition-first action models}, $\FGAM$ is a tree-like single-coalition-first action model. It is easy to show that $\FGAM$ is $z$-representable by a tree-like single-coalition-first neighborhood model.

\end{proof}

\begin{theorem}
\label{theorem:representation tree-like single-coalition-first neighborhood models TO tree-like single-coalition-first action models}

Every tree-like single-coalition-first neighborhood model $z$-represents a tree-like grand-coalition-first action model.

\end{theorem}

\begin{proof}~

Let $\SNM = (\FST, \{\Fnei_\FCC \mid \FCC \subseteq \FAG\}, \Flab)$ be a tree-like single-coalition-first neighborhood model.
Define a single-coalition-first action model $\SAM = (\FST, \FAC, \{\Fout_\FCC \mid \FCC \subseteq \FAG\}, \{\Faja_\FCC \mid \FCC \subseteq \FAG\}, \Flab)$ as in the proof for Theorem \ref{theorem:representation single-coalition-first neighborhood models TO single-coalition-first action models}. As shown there, $\SNM$ $z$-represents $\SAM$.
It is easy to check that $\SAM$ is tree-like. By Theorem \ref{theorem:implication}, $\SAM$ is a grand-coalition-first action model.

\end{proof}

The following result follows from Theorems \ref{theorem:representation tree-like grand-coalition-first action models TO tree-like single-coalition-first neighborhood models}, \ref{theorem:representation tree-like single-coalition-first neighborhood models TO tree-like single-coalition-first action models}, and \ref{theorem:X iff X}.

\begin{theorem}[Representation of the class of tree-like grand-coalition-first action $\FXX$-models by the class of tree-like single-coalition-first neighborhood $\FXX$-models]
\label{theorem:Representation of tree-like grand-coalition-first action models by tree-like single-coalition-first neighborhood models}

For every $\FXX \in \FES$, the class of tree-like grand-coalition-first action $\FXX$-models is $z$-representable by the class of tree-like single-coalition-first neighborhood $\FXX$-models.

\end{theorem}

\section{Each of those eight coalition logics is determined by these six kinds of models, too}
\label{section:Each of the eight coalition logics is determined by the six kinds of models, too}

In the following, for every finite nonempty sequence $\pi$, we use $\pi^l$ to denote the last element of $\pi$.

\begin{definition}[Unravelling of pointed grand-coalition-first action models]

Let $(\FGAM, s)$ be a pointed grand-coalition-first action model, where $\FGAM = (\FST, \FAC, \{\Fout_\FCC \mid \FCC \subseteq \FAG\}, \{\Fav_\FCC \mid \FCC \subseteq \FAG\}, \Flab)$.

Define a grand-coalition-first action model $\FGAM' = (\FST', \FAC, \{\Fout'_\FCC \mid \FCC \subseteq \FAG\}, \{\Fav'_\FCC \mid \FCC \subseteq \FAG\}, \Flab')$ as follows:
\begin{itemize}

\item

$\FST'= \FST_0' \cup \FST_1' \cup \FST_2' \cup \dots$, where:
\begin{itemize}

\item
$\FST_0'= \{s\}$;

\item

$\FST_{k+1}' = \{\pi-\sigma_\FAG-u \mid \pi \in \FST_{k}', \sigma_\FAG \in \FJA_\FAG, \text{and } u \in \Fout_\FAG(\pi^l,\sigma_\FAG)\} $.

\end{itemize}

\item

\begin{itemize}

\item

for every $\pi \in \FST'$ and $\sigma_\FAG \in \FJA_\FAG$,
$\Fout'_\FAG (\pi, \sigma_\FAG) = \{ \pi-\sigma_\FAG-u \mid u \in \Fout_\FAG ( \pi^l, \sigma_\FAG) \}$.

\item 

for all $\FCC \subseteq \FAG$, $\pi \in \FST'$ and $\sigma_\FCC \in \FJA_\FCC$: $\Fout'_\FCC (\pi, \ja{\FCC}) = \bigcup \{\Fout'_\FAG (\pi, \ja{\FAG}) \mid \ja{\FAG} \in \FJA_\FAG \text{ and } \ja{\FCC} \subseteq \ja{\FAG}\}$.

\end{itemize}

\item 

for every $\FCC \subseteq \FAG$ and $\pi \in \FST'$, $\Faja'_\FCC (\pi) = \{\ja{\FCC} \in \FJA_\FCC \mid \Fout'_\FCC (\pi, \ja{\FCC}) \neq \emptyset\}$.

\item

for every $\pi \in \FST'$, $\Flab'(\pi) = \{ p \in \FAP |$ $p \in \Flab(\pi^l)\}$.

\end{itemize}

$(\FGAM', s)$ is called the \Fdefs{unravelling} of $(\FGAM, s)$.

\end{definition}

It is easy to check that $(\FGAM', s)$ is tree-like.
This notion is a variation of the notion of unravelling used throughout modal logic. The slight difference is that, in the unravelled pointed model defined above, states are sequences consisting of states (of the original model) and action profiles, rather than just states (of the original model).

\begin{lemma}

Let $(\FGAM, s)$ be a pointed grand-coalition-first action model and $(\FGAM', s)$ be its \Fdefs{unravelling}, where $\FGAM = (\FST, \FAC, \{\Fout_\FCC \mid \FCC \subseteq \FAG\}, \{\Fav_\FCC \mid \FCC \subseteq \FAG\}, \Flab)$ and $\FGAM' = (\FST', \FAC, \{\Fout'_\FCC \mid \FCC \subseteq \FAG\}, \{\Fav'_\FCC \mid \FCC \subseteq \FAG\}, \Flab)$.

Then:
\begin{enumerate}[label=(\arabic*),leftmargin=3.33em]

\item 

for every $\pi \in \FST'$, $\FCC \subseteq \FAG$, $\sigma_\FCC \in \FJA_\FCC$, and $\pi' \in \FST'$, if $\pi' \in \Fout'_\FCC (\pi,\sigma_\FCC)$, then $\pi'^l \in \Fout_\FCC (\pi^l,\sigma_\FCC)$;

\item 

for every $\pi \in \FST'$, $\FCC \subseteq \FAG$, $\sigma_\FCC \in \FJA_\FCC$, $t \in \FST$, if $t \in \Fout_\FCC (\pi^l,\sigma_\FCC)$, then there is $\pi' \in \FST'$ such that $\pi'^l = t$ and $\pi' \in \Fout'_\FCC (\pi,\sigma_\FCC)$.

\end{enumerate}

\end{lemma}

\begin{proof}~

\begin{enumerate}[label=(\arabic*),leftmargin=3.33em]

\item 

Let $\pi \in \FST'$, $\FCC \subseteq \FAG$, $\sigma_\FCC \in \FJA_\FCC$, and $\pi' \in \FST'$. Assume $\pi' \in \Fout'_\FCC (\pi, \sigma_\FCC)$. Then, $\pi' \in \Fout'_\FAG (\pi, \sigma_\FAG)$ for some $\sigma_\FAG \in \FJA_\FAG$ such that $\sigma_\FCC \subseteq \sigma_\FAG$.
Then, $\pi' = \pi-\sigma_\FAG-u$ for some $u \in \Fout_\FAG(\pi^l, \sigma_\FAG)$. Note $\Fout_\FAG(\pi^l, \sigma_\FAG) \subseteq \Fout_\FCC(\pi^l, \sigma_\FCC)$. Then,  $u \in \Fout_\FCC (\pi^l, \sigma_\FCC)$, that is, $\pi'^l \in \Fout_\FCC (\pi^l, \sigma_\FCC)$.

\item

Let $\pi \in \FST'$, $\FCC \subseteq \FAG$, $\sigma_\FCC \in \FJA_\FCC$, and $t \in \FST$. Assume $t \in \Fout_\FCC (\pi^l, \sigma_\FCC)$. Then, $t \in \Fout_\FAG (\pi^l, \sigma_\FAG)$ for some $\sigma_\FAG \in \FJA_\FAG$ such that $\sigma_\FCC \subseteq \sigma_\FAG$. 
Then, there is $\pi' \in \FST'$ such that $\pi' = \pi-\sigma_\FAG-t$ and $\pi' \in \Fout'_\FAG (\pi,\sigma_\FAG)$. Note $\Fout'_\FAG(\pi, \sigma_\FAG) \subseteq \Fout'_\FCC(\pi, \sigma_\FCC)$. Then, $\pi' \in \Fout'_\FCC (\pi, \sigma_\FCC)$.

\end{enumerate}

\end{proof}

\begin{lemma}
\label{lemma:unravelling X model}

Let $\FXX \in \FES$, $(\FGAM, s)$ be a pointed grand-coalition-first action $\FXX$-model, and $(\FGAM', s)$ be the \Fdefs{unravelling}, where $\FGAM = (\FST, \FAC, \{\Fout_\FCC \mid \FCC \subseteq \FAG\}, \{\Fav_\FCC \mid \FCC \subseteq \FAG\}, \Flab)$ and $\FGAM' = (\FST', \FAC, \{\Fout'_\FCC \mid \FCC \subseteq \FAG\}, \{\Fav'_\FCC \mid \FCC \subseteq \FAG\}, \Flab)$.

Then $(\FGAM', s)$ is a pointed grand-coalition-first action $\FXX$-model.

\end{lemma}

\begin{proof}~

Assume $\FGAM$ is serial. Let $\pi \in \FST'$. Note there is $\sigma_\FAG$ such that $\Fout_\FAG (\pi^l, \sigma_\FAG)$ $ \neq \emptyset$. Let $u \in \Fout_\FAG (\pi^l, \sigma_\FAG)$. Then, $\pi - \sigma_\FAG - u \in \Fout'_\FAG (\pi, \sigma_\FAG)$. Then, $\Fout'_\FAG (\pi, \sigma_\FAG) \neq \emptyset$. Then, $\FGAM'$ is serial.

Assume $\FGAM$ is independent. Let $\pi \in \FST'$, $\FCC, \FDD \subseteq \FAG$ such that $\FCC \cap \FDD = \emptyset$, $\sigma_\FCC \in \Faja'_\FCC(\pi)$, and $\sigma_\FDD \in \Faja'_\FDD(\pi)$.
To show $\FGAM'$ is independent, it suffices to show $\sigma_\FCC \cup \sigma_\FDD \in \Faja'_{\FCC \cup \FDD}(\pi)$.
Note $\Fout'_{\FCC}(\pi, \sigma_\FCC) \neq \emptyset$ and $\Fout'_{\FDD}(\pi, \sigma_\FDD) \neq \emptyset$. By the claim (1), $\Fout_{\FCC}(\pi^l, \sigma_\FCC) \neq \emptyset$ and $\Fout_{\FDD}(\pi^l, \sigma_\FDD) \neq \emptyset$. Then $\sigma_\FCC \in \Faja_\FCC(\pi^l)$ and $\sigma_\FDD \in \Faja_\FDD (\pi^l)$.
Then, $\sigma_\FCC \cup \sigma_\FDD \in \Faja_{\FCC \cup \FDD}(\pi^l)$. Then, $\Fout_{\FCC \cup \FDD}(\pi^l, \sigma_\FCC \cup \sigma_\FDD) \neq \emptyset$. By the claim (2), $\Fout'_{\FCC \cup \FDD}(\pi, \sigma_\FCC \cup \sigma_\FDD) \neq \emptyset$.
Then, $\sigma_\FCC \cup \sigma_\FDD \in \Faja'_{\FCC \cup \FDD}(\pi)$.

Assume $\FGAM$ is deterministic. Assume $\FGAM'$ is not deterministic. Then, there is $\pi \in \FST'$ and $\sigma_\FAG \in \Faja'_\FAG(\pi)$ such that $\Fout'_\FAG(\pi, \sigma_\FAG)$ is not a singleton. Let $\pi', \pi'' \in \Fout'_\FAG (\pi, \sigma_\FAG)$ such that $\pi' \neq \pi''$.
Let $\pi' = \pi - \sigma_\FAG - t'$ and $\pi'' = \pi - \sigma_\FAG - t''$. Note $t' \neq t''$. Then, $t', t'' \in \Fout_\FAG(\pi^l, \sigma_\FAG)$. We have a contradiction.

\end{proof}

\begin{lemma}
\label{lemma:unravelling invariant}

Let $\FXX \in \FES$, $(\FGAM, s)$ be a pointed grand-coalition-first action model, and $(\FGAM', s)$ be the \Fdefs{unravelling} of $(\FGAM, s)$, where $\FGAM = (\FST, \FAC, \{\Fout_\FCC \mid \FCC \subseteq \FAG\}, \{\Fav_\FCC \mid \FCC \subseteq \FAG\}, \Flab)$ and $\FGAM' = (\FST', \FAC, \{\Fout'_\FCC \mid \FCC \subseteq \FAG\}, \{\Fav'_\FCC \mid \FCC \subseteq \FAG\}, \Flab)$.

Then for all $\pi \in \FST'$ and $\phi \in \Phi$, $\FGAM', \pi \Vdash \phi$ if and only if $\FGAM, \pi^l \Vdash \phi$.

\end{lemma}

\begin{proof}~

We put an induction on $\phi$. We consider only the case $\phi = \Fclo{\FCC} \psi$ and skip others.

Assume $\FGAM', \pi \Vdash \Fclo{\FCC} \psi$. Then, there is $\sigma_\FCC \in \Fav'_\FCC (\pi)$ such that for all $\pi' \in \Fout'_\FCC(\pi, \sigma_\FCC)$, $\FGAM', \pi' \Vdash \psi$.
We want to show $\sigma_\FCC \in \Fav_\FCC (\pi^l)$ and for all $t \in\Fout_\FCC(\pi^l, \sigma_\FCC)$, $\FGAM, t \Vdash \psi$.
Note $\Fout'_\FCC (\pi, \sigma_\FCC) \neq \emptyset$.
By the claim (1), $\Fout_{\FCC}(\pi^l,\sigma_\FCC) \neq \emptyset$. Then, $\sigma_\FCC \in \Faja_\FCC (\pi^l)$.
Let $t \in \Fout_\FCC(\pi^l, \sigma_\FCC)$. By the claim (2), there is a $\pi' \in \FST'$ such that $\pi'^l = t$ and $\pi' \in \Fout'_\FCC (\pi, \sigma_\FCC)$.
Note $\FGAM', \pi' \Vdash \psi$. By the inductive hypothesis, $\FGAM, t \Vdash \psi$.

Assume $\FGAM, \pi^l \Vdash \Fclo{\FCC} \psi$. Then, there is $\sigma_\FCC \in \Fav_\FCC (\pi^l) $ such that for all $t \in \Fout_\FCC(\pi^l, \sigma_\FCC)$, $\FGAM, t \Vdash \psi$.
We want to show $\sigma_\FCC \in \Fav'_\FCC (\pi)$ and for all $\pi' \in \Fout'_\FCC(\pi, \sigma_\FCC)$, $\FGAM', \pi' \Vdash \psi$.
Note $\Fout_\FCC (\pi^l, \sigma_\FCC) \neq \emptyset$. By the claim (2), $\Fout'_\FCC (\pi, \sigma_\FCC) \neq \emptyset$. Then $\sigma_\FCC \in \Fav'_\FCC (\pi)$.
Let $\pi' \in \Fout'_\FCC(\pi, \sigma_\FCC)$. By the claim (1), $\pi'^l \in \Fout_\FCC(\pi^l, \sigma_\FCC)$. Note $\GAM, \pi'^l \Vdash \psi$. By the inductive hypothesis, $\FGAM', \pi' \Vdash \psi$.

\end{proof}

From Lemmas \ref{lemma:unravelling X model} and \ref{lemma:unravelling invariant}, the following result follows:

\begin{theorem}
\label{theorem:vdash}

For every $\FXX \in \FES$ and pointed grand-coalition-first action $\FXX$-model $(\FGAM, s)$, there is a pointed tree-like grand-coalition-first action $\FXX$-model $(\FGAM', s')$ such that for all $\phi \in \Phi$, $\FGAM, s \Vdash \phi$ if and only if $\FGAM', s' \Vdash \phi$.

\end{theorem}

\medskip

\begin{theorem}

For all $\FXX \in \FES$ and $\phi \in \Phi$, the following statements are equivalent:
\begin{enumerate}[label=(\arabic*),leftmargin=3.33em,start=0]

\item

$\phi$ is valid with respect to the class of \emph{grand-coalition-first action $\FXX$-models};

\item 

$\phi$ is valid with respect to the class of \emph{single-coalition-first action $\FXX$-models};

\item 

$\phi$ is valid with respect to the class of \emph{single-coalition-first neighborhood $\FXX$-models};

\item 

$\phi$ is valid with respect to the class of \emph{clear grand-coalition-first action $\FXX$-models};

\item 

$\phi$ is valid with respect to the class of \emph{clear single-coalition-first neighborhood $\FXX$-models};

\item 

$\phi$ is valid with respect to the class of \emph{tree-like grand-coalition-first action $\FXX$-models};

\item 

$\phi$ is valid with respect to the class of \emph{tree-like single-coalition-first neighborhood $\FXX$-models}.

\end{enumerate}

\end{theorem}

\begin{proof}~

By Theorem \ref{theorem:Representation of single-coalition-first action models by single-coalition-first neighborhood models} and Theorem \ref{theorem:classes representable}, (1) and (2) are equivalent.
By Theorem \ref{theorem:Representation of clear grand-coalition-first action models by clear single-coalition-first neighborhood models} and Theorem \ref{theorem:classes representable}, (3) and (4) are equivalent.
By Theorem \ref{theorem:Representation of tree-like grand-coalition-first action models by tree-like single-coalition-first neighborhood models} and Theorem \ref{theorem:classes representable}, (5) and (6) are equivalent.
It suffices to show that (0), (1), (3), and (5) are equivalent.
By Theorem \ref{theorem:implication}, (0) implies (1).
By Theorem \ref{theorem:clear grand-coalition-first action models are single-coalition-first action models}, (1) implies (3).
By Theorem \ref{theorem:tree-like implies clear}, (3) implies (5).
By Theorem \ref{theorem:vdash}, (5) implies (0).

\end{proof}

\section{Conclusion}
\label{section:Conclusion}

Li and Ju \cite{li_minimal_2025,li_completeness_2024} proposed eight coalition logics based on grand-coalition-first action models. In this work, we show that each of them is also determined by six other kinds of models: single-coalition-first action models, single-coalition-first actual neighborhood models, clear grand-coalition-first action models, clear single-coalition-first actual neighborhood models, tree-like grand-coalition-first action models, and tree-like single-coalition-first actual neighborhood models.

Among the six kinds of models, except for tree-like grand-coalition-first action models and tree-like single-coalition-first actual neighborhood models, which mostly make technical sense, the other four kinds of models make good sense for representing coalitional powers. The language of coalition logic is not very expressive and cannot distinguish between them. Further work worth doing is to consider richer languages.

The operator $\Fclo{\FCC} \phi$ of coalition logic is a bundled operator: it implicitly has an existential quantifier over actions, a universal operator over possible outcomes of actions, and a next-time operator. What follows are three ideas.
The first is to consider the language of Alternating-time Temporal Logic $\FATL$ and to see whether we could have similar results. The language of $\FATL$ is an extension of the language of coalition logic with the operators $\Fclo{\FCC} \mathtt{G} \phi$ (\emph{$\FCC$ has a strategy to ensure $\phi$ forever}) and $\Fclo{\FCC} \phi \mathtt{U} \psi$ (\emph{$\FCC$ has a strategy to ensure $\phi$ until $\psi$}).
The second idea is to separate the quantifier over actions in the meaning of $\Fclo{\FCC} \phi$. Logics following this idea would have connections to logics of ``bring-it-about'' \cite{porn_logic_1970, elgesem_action_1993}.
The third idea is to consider languages of STIT (see to it that) logics \cite{belnap_seeing_1988}, which are temporal extensions of logics of ``bring-it-about''.

\subsection*{Acknowledgments}

Thanks go to Thomas \r{A}gotnes, Valentin Goranko, the three anonymous reviewers and the audience of the 6th International Conference on Logic and Argumentation, and the audience of a logic seminar at Beijing Normal University.

\bibliography{Strategy-reasoning,Strategy-reasoning-special}
\bibliographystyle{alpha}

\end{document}